\documentclass[letterpaper,journal]{IEEEtran}
\usepackage{amsmath,amsfonts,amssymb}
\usepackage{algorithm}
\usepackage{algpseudocode}
\usepackage{array}
\usepackage[caption=false,font=normalsize,labelfont=sf,textfont=sf]{subfig}
\usepackage{textcomp}
\usepackage{stfloats}
\usepackage{url}
\usepackage{verbatim}
\usepackage{graphicx}
\usepackage{cite}

\usepackage[hidelinks]{hyperref}
\usepackage{orcidlink}

\usepackage{listings}
\usepackage{mdframed}
\usepackage{tabularray}
\UseTblrLibrary{booktabs}
\usepackage{setspace}
\usepackage{comment}
\usepackage{float}
\usepackage{xcolor}
\usepackage{tcolorbox}
\usepackage{lipsum}
\usepackage{bm}
\usepackage{makecell}
\usepackage{multirow}
\usepackage{fancyhdr}
\usepackage{enumerate}
\usepackage[figuresright]{rotating}
\usepackage{colortbl}
\usepackage{ifthen}

\mdfdefinestyle{mystyle}{
  linewidth=1pt,
  linecolor=black,
  backgroundcolor=white,
}

\newcommand{\tabincell}[2]{\begin{tabular}{@{}#1@{}}#2\end{tabular}}
\newcolumntype{C}[1]{>{\centering\arraybackslash}p{#1}}

\definecolor{texgray}{RGB}{244, 245, 246}
\definecolor{texrightgray}{RGB}{99, 99, 99}
\definecolor{green}{RGB}{214, 230, 189}
\definecolor{leftgreen}{RGB}{19, 138, 7}
\definecolor{gray}{RGB}{200,200,200}
\definecolor{highlightcolor}{RGB}{255,251,204}

\newmdenv[
  leftline=false,
  rightline=true,
  topline=false,
  bottomline=false,
  backgroundcolor=texgray,
  linecolor=texrightgray,
  linewidth=2pt,
  skipabove=5pt
]{myrightline}

\newmdenv[
  leftline=true,
  rightline=false,
  topline=false,
  bottomline=false,
  backgroundcolor=green!20,
  linecolor=leftgreen,
  linewidth=2pt,
  skipabove=5pt
]{myleftline}

\newboolean{showcomments}
\setboolean{showcomments}{true}
\ifthenelse{\boolean{showcomments}}
  {
    \newcommand{\mynote}[2]{
        \fbox{\bfseries\sffamily\scriptsize#1}
        {
          \small$\blacktriangleright$
          \textsf{
            \emph{
              \ifthenelse{\equal{#1}{RS}}{\textcolor{red}{#2}}{
                \ifthenelse{\equal{#1}{Brygon}}{\textcolor{blue}{#2}}{
                  \ifthenelse{\equal{#1}{Zheng}}{\textcolor{orange}{#2}}{
                    \textcolor{black}{#2}
                  }
                }
              }
            }
          }
          $\blacktriangleleft$
        }
    }
  }
  {
    \newcommand{\mynote}[2]{}
  }

\lstdefinelanguage{JSON}{
    basicstyle=\ttfamily,
    breaklines=true,
    showstringspaces=false,
    keywordstyle=\color{blue}\bfseries,
    morestring=[b]",
    morecomment=[l]{//},
    morecomment=[s]{/*}{*/},
    morekeywords={mutant_id,is_interference_mutant,reason},
    literate={\"}{{\"}}1,
}

\lstnewenvironment{code}[1][]
{\lstset{
  language=C,
  escapeinside={(*@}{@*)},
  basicstyle=\scriptsize\ttfamily,
  breaklines=true,
  numbers=none,
  numbersep=5pt,
  xleftmargin=10pt,
  showstringspaces=false,
}}
{}

\begin{document}

\newtheorem{theorem}{Theorem}
\newtheorem{definition}{Definition}
\newtheorem{corollary}{Corollary} 
\newtheorem{proof}{Proof}
\newtheorem{lemma}{Lemma}

\title{FLIMs: Fault Localization Interference Mutants, Definition, Recognition and Mitigation}

\author{Hengyuan Liu\orcidlink{0000-0002-5884-2089},
        Zheng Li\orcidlink{0000-0002-3938-7033},~\IEEEmembership{Member,~IEEE},
        Donghua Wang,
        Yankai Wu,
        Xiang Chen\orcidlink{0000-0002-1180-3891},~\IEEEmembership{Member,~IEEE},
        and~Yong Liu\orcidlink{0000-0003-1754-3039},~\IEEEmembership{Member,~IEEE}
        \thanks{Manuscript received XXX XX, 2025; revised XXX XX, 2025. \textit{(Hengyuan Liu and Zheng Li contributed equally to this work; Corresponding author: Yong Liu.)}}
        \thanks{Hengyuan Liu, Zheng Li, Donghua Wang, Yankai Wu, and Yong Liu are with the College of Information Science and Technology, Beijing University of Chemical Technology, Beijing, China. (e-mail: lhywandm@163.com; lizheng@mail.buct.edu.cn; 2023200825@buct.edu.cn; 1074154081@qq.com; lyong@mail.buct.edu.cn).}
        \thanks{Xiang Chen is with the School of Artificial Intelligence and Computer Science, Nantong University, Nantong, China (e-mail: xchencs@ntu.edu.cn).}
        \thanks{Digital Object Identifier XX.XXXX/TSE.XXXX.XXXXXXX}}

\markboth{IEEE Transactions on Software Engineering,~Vol.~XX, No.~X, XXX~2025}
{Liu \MakeLowercase{\textit{et al.}}: FLIMs: Fault Localization Interference Mutants, Definition, Recognition and Mitigation}

\maketitle

\begin{abstract}

Mutation-based Fault Localization (MBFL) has been widely explored for automated software debugging, leveraging artificial mutants to identify faulty code entities. However, MBFL faces significant challenges due to interference mutants generated from non-faulty code entities but can be killed by failing tests. These mutants mimic the test sensitivity behaviors of real faulty code entities and weaken the effectiveness of fault localization.
To address this challenge, we introduce the concept of Fault Localization Interference Mutants (FLIMs) and conduct a theoretical analysis based on the Reachability, Infection, Propagation, and Revealability (RIPR) model, identifying four distinct interference causes.
Building on this theoretical analysis, we propose a novel approach to semantically recognize and mitigate FLIMs using LLM-based semantic analysis, enhanced by fine-tuning techniques and confidence estimation strategies to address LLM output instability. The recognized FLIMs are then mitigated by refining the suspiciousness scores calculated from MBFL techniques.
We integrate FLIM recognition and mitigation into the MBFL workflow, developing MBFL-FLIM, a fault localization framework that enhances MBFL's effectiveness by reducing misleading interference while preserving real fault-revealing information. Our empirical experiments on the Defects4J benchmark with 395 program versions using eight LLMs demonstrate MBFL-FLIM's superiority over traditional SBFL and MBFL methods, advanced dynamic feature-based approaches, and recent LLM-based fault localization techniques. 
Specifically, MBFL-FLIM achieves an average absolute improvement of 44 faults in the Top-1 metric, representing a significant enhancement over baseline methods. Further evaluation confirms MBFL-FLIM's robust performance in multi-fault scenarios, with ablation experiments validating the contributions of the fine-tuning and confidence estimation components.

\end{abstract}


\section{Introduction}
\label{sec:Introduction}

\IEEEPARstart{S}{oftware} faults are inevitable in modern software systems, making fault localization a critical activity in software engineering~\cite{wong2016survey}. Traditional debugging approaches rely heavily on manual inspection and domain expertise, which can be time-consuming and error-prone, especially for large-scale software systems~\cite{pearson2017evaluating}. Automated fault localization techniques have emerged as essential tools to assist developers in localizing the faulty code entities in the software, thereby reducing debugging time and improving software quality.

Mutation-based Fault Localization (MBFL) is a widely studied automated fault localization technique with a theoretical foundation~\cite{papadakis2015metallaxis, moon2014ask}.
MBFL adopts mutation testing whereby mutation operators are applied to generate mutants and the test suite is executed on each mutant to obtain a kill matrix recording the test result changes of mutants~\cite{chekam2016assessing, wang2025systematic}.
Under the basic assumption that mutants of faulty code entities have higher failing-test kill rates and lower passing-test kill rates, the kill matrix is mapped to each mutant's suspiciousness and aggregated to produce a ranked list of program entities for localization~\cite{wang2025systematic, li2017transforming}.
The competent programmer hypothesis~\cite{jia2010analysis} and the coupling effect hypothesis~\cite{offutt1992investigation} ensure that mutants generated by mutation testing can simulate real complex faults to some extent, while the suspiciousness formula design assumption guides the mapping from mutation analysis results to program entity suspiciousness, making it quantifiable~\cite{papadakis2015metallaxis, moon2014ask}.

MBFL faces significant challenges that limit its effectiveness in practice, primarily due to the high computational cost of mutant generation and execution. Recent research has observed a more subtle but critical issue: the existence of interference phenomena that can mislead fault localization techniques~\cite{liu2024delta4ms, du2022improving, jang2022hotfuz, kim2023learning}. Specifically, certain mutants generated from non-faulty code entities (which should not be killed) are killed by failing test cases and exhibit behaviors similar to those of mutants from real faulty entities. 
Such mutants may exhibit high killing rates by failing test cases, resulting in inflated suspiciousness scores for non-faulty code entities and consequently degrading the accuracy of fault localization rankings.
Empirical observations reveal that such interference mutants can distort the effectiveness of MBFL by creating false positive indications that compete with real fault-revealing indications from actual faulty locations~\cite{liu2025scope}.

The complexity of the interference phenomena in fault localization makes it a challenge to understand and resolve~\cite{liu2024delta4ms, du2022improving, jang2022hotfuz, kim2023learning}. 
Interference phenomena were investigated in mutation testing, particularly in the context of test suite quality assessment~\cite{tian2024large, schuler2013covering, garg2022cerebro}, where equivalent mutants and subsuming mutants are well-characterized interference factors that primarily affect mutation score accuracy. 
Unlike mutation testing, MBFL calculates the suspiciousness of each mutant and ranks them, making it more sensitive to interference mutants. Existing research~\cite{liu2024delta4ms, du2022improving, jang2022hotfuz, kim2023learning} on MBFL interference phenomena has primarily focused on statistically measuring the effect or analysing the coupling relationship of interference mutants.
However, due to the complexity of the mutant-injected faults and the program's original fault, recognizing interference mutants requires deep semantic understanding of how mutant modifications interact with program semantics, test case characteristics, execution contexts, and suspiciousness calculation assumptions.
Real-world interference mutants often exhibit characteristics from multiple root causes simultaneously, with subtle disruptive effects that require sophisticated reasoning to distinguish from fault-revealing behavior. Additionally, practical software projects generate thousands of mutants, making manual analysis infeasible and demanding automated approaches capable of large-scale semantic analysis. 
Therefore, it is beneficial to systematically understand the MBFL interference mutants and develop more automated recognition and mitigation methods with advanced semantic analysis techniques like Large Language Models (LLMs), considering the complex interaction between mutants and program semantics.

To systematically understand MBFL interference mutants, we introduce the concept of Fault Localization Interference Mutants (FLIMs) and conduct a theoretical analysis via the Reachability, Infection, Propagation, and Revealability (RIPR) model~\cite{li2016test}. Four distinct interference root causes have been identified that demonstrate the manner in which such mutants interfere with suspiciousness calculation in MBFL. This theoretical analysis provides a systematic framework for understanding why certain mutants interfere with fault localization and establishes the basis for developing recognition and mitigation strategies by revealing the semantic complexity and contextual dependencies.

Furthermore, to alleviate the interference phenomena in MBFL, we propose a novel approach to semantically recognize and mitigate FLIMs using LLM-based semantic analysis. Specifically, FLIM recognition employs LLM-based semantic analysis to recognize FLIMs through code context and dynamic execution result understanding, incorporates specialized fine-tuning technique to adapt LLMs for domain-specific FLIM recognition tasks, and implements confidence estimation algorithms to enhance recognition reliability and stability. 
FLIM mitigation refines the suspiciousness scores computed by traditional MBFL techniques utilizing the FLIM recognition results. We develop MBFL-FLIM (FLIM Recognition-based Fault Localization), a fault localization framework that integrates FLIM recognition and mitigation to alleviate the impact of FLIMs on fault localization effectiveness.

To evaluate the effectiveness of MBFL-FLIM, we conduct extensive experiments using the Defects4J v1.2.0 benchmark with 395 buggy programs across six open-source Java projects. 
Our experimental design compares MBFL-FLIM against baselines including SBFL~\cite{abreu2006evaluation}, MBFL~\cite{papadakis2015metallaxis}, dynamic feature-based baselines (i.e., Delta4Ms~\cite{liu2024delta4ms, du2022improving}, BLMu~\cite{wang2025systematic}, SmartFL~\cite{wu2025smartfl}), and recent LLM-based fault localization techniques (i.e., LLMAO~\cite{yang2024large}, LLMFL~\cite{wu2023large}). 
The experimental results demonstrate that MBFL-FLIM achieves substantial improvements across all evaluation metrics: 96 at Top-1, 133 at Top-3, 171 at Top-5 and a MFR of 19.74, significantly outperforming all baseline methods. Compared to baseline methods, MBFL-FLIM demonstrates significant absolute improvements: an average increase of 44 faults for Top-1 (from an average baseline of 52 to 96), an increase of 30 faults for Top-3 (from an average baseline of 103 to 133), and an increase of 38 faults for Top-5 (from an average baseline of 133 to 171).
In multi-fault scenarios, MBFL-FLIM maintains robust performance. 
The ablation study reveals that both fine-tuning and confidence estimation components contribute to the overall performance, with the Principal Component Analysis (PCA) confidence estimation algorithm achieving optimal results.

We summarize our main contributions as follows:
\begin{itemize}
    \item \textbf{FLIMs: Formal Concept and Root Cause Analysis:} We introduce the formal concept of Fault Localization Interference Mutants (FLIMs) and propose a theoretical analysis based on the RIPR model. Four distinct interference root causes (reach-based, infection-based, propagation-based, and revealability-based interference) are identified that systematically explain the root causes of mutant interference in fault localization.
    
    \item \textbf{LLM-based FLIM Recognition:} We develop a novel LLM-based approach for FLIM recognition that leverages semantic code analysis capabilities. The method incorporates LLM fine-tuning technique for domain-specific FLIM recognition tasks and implements confidence estimation strategies to enhance recognition reliability and overcome the instability of LLM output.
    
    \item \textbf{FLIM Mitigation in MBFL:} We propose MBFL-FLIM, a fault localization framework that systematically mitigates FLIM interference in MBFL. The framework integrates FLIM recognition with suspiciousness refinement to mitigate the influence of misleading interference while preserving real fault-revealing information, thereby improving the effectiveness of MBFL.
    
    \item \textbf{Comprehensive Empirical Study:} We conduct experiments on the Defects4J benchmark with 395 program versions using eight LLMs from multiple series. The experimental results demonstrate that MBFL-FLIM outperforms existing SBFL, MBFL, dynamic feature-based, and LLM-based fault localization approaches across multiple evaluation metrics.
\end{itemize}

\section{Background}
\label{sec:Background}

\subsection{Mutation-Based Fault Localization}

\subsubsection{Theoretical Foundations of MBFL}

MBFL represents a systematic approach to software debugging that builds upon three theoretical foundations~\cite{papadakis2015metallaxis, moon2014ask}. 
At its core, MBFL operates on the principle that artificial faults, introduced through mutation operators, can simulate real programming errors and guide developers toward actual fault locations.

The foundation of MBFL rests on \emph{mutation analysis}, which provides the fundamental basis for fault simulation~\cite{papadakis2015metallaxis}. Through syntactic transformations governed by predefined rules (e.g., changing $a < b$ to $a > b$), mutation operators generate program mutants~\cite{papadakis2015metallaxis, kooli2017computing}. Each mutant $m$ represents a potential fault pattern, and its behavior under test execution determines whether it is killed (detected by at least one test case) or remains alive (undetected)~\cite{debroy2014combining, fan2023sgs}.

The validity of using artificial mutants to locate real faults is justified by the \emph{competent programmer hypothesis}~\cite{jia2010analysis}, which assumes that programmers produce code that is syntactically close to the correct version. Under this assumption, real faults typically manifest as small syntactic deviations from the intended implementation, making them similar to the artificial faults created by mutation operators. This bridge between artificial and real faults legitimizes the use of mutant behavior patterns as indicators of fault locations.

The \emph{coupling effect hypothesis} strengthens MBFL's foundation by establishing the relationship between simple and complex faults~\cite{offutt1992investigation}. This hypothesis asserts that test cases capable of detecting simple faults (represented by mutants) are also likely to reveal complex, real-world faults. Consequently, the interaction patterns between test cases and mutants provide diagnostic information about fault characteristics, enabling MBFL to infer fault locations through analysis of mutant kill patterns.

Based on these three foundations, MBFL techniques can compute suspiciousness scores that quantify the possibility of each program entity containing a fault by analyzing the behavior of mutants under passing and failing test cases~\cite{papadakis2015metallaxis, moon2014ask}. The effectiveness of MBFL fundamentally relies on its suspiciousness formula design assumption that distinguishes mutants generated from faulty statements versus those from non-faulty statements based on their test sensitivity profiles. Specifically, mutants from faulty statements should exhibit significantly higher kill rates when exposed to failing tests compared to mutants from non-faulty statements~\cite{papadakis2015metallaxis, moon2014ask}, as the former are more likely to interact with the error propagation pathways that lead to observable failures. This assumption forms the basis underlying suspiciousness calculation and ranking in MBFL approaches, which has proven effective in practice and established MBFL as a reliable methodology for software debugging.

\subsubsection{The Process of MBFL}

The MBFL process consists of five main steps that transform mutation analysis results into ranked suspiciousness scores indicating fault localization~\cite{papadakis2015metallaxis}.

\textbf{Step 1: Test Execution and Classification.} The MBFL process begins by executing the original program $P$ using the test suite $T = \{t_1, t_2, \ldots, t_{|T|}\}$. The execution function $\mathcal{O}(P, T)$ records test outcomes, while the coverage function $\mathcal{C}(P, T)$ captures which program entities are executed. The test suite $T$ is partitioned into two subsets: $T_p = \{t \in T : \mathcal{O}(P, t) = \text{pass}\}$ representing passing tests and $T_f = \{t \in T : \mathcal{O}(P, t) = \text{fail}\}$ representing failing tests. The set of suspicious program entities $\mathcal{E}$ is determined by the coverage of failing tests: $\mathcal{E} = \{e : e \in \mathcal{C}(P, T_f)\}$, as only entities executed by failing tests can potentially contain faults.

\textbf{Step 2: Mutant Generation and Execution.} Mutation operators $\mathcal{M} = \{\mu_1, \mu_2, \ldots, \mu_k\}$ are applied to generate mutants for the suspicious entities identified in Step 1. The mutant set is defined as $\mathcal{M}(P) = \bigcup_{\mu \in \mathcal{M}} \mu(P, \mathcal{E})$, where each $\mu(P, \mathcal{E})$ produces mutants through syntactic transformations applied to entities in $\mathcal{E}$. Each mutant $m_i \in \mathcal{M}(P)$ represents a syntactic modification that simulates potential faults at a program location. All generated mutants are executed against the test suite $T$, with execution outcomes recorded as $\mathcal{O}(m_i, T)$ and the mutant location information captured as $\text{loc}(m_i)$.

\textbf{Step 3: Mutant Killing Analysis.} This step determines whether each mutant is killed by comparing its execution behavior with the original program. Two killing criteria are commonly used. \emph{Strong Killing}~\cite{harman2011Strong} considers a mutant killed if its execution changes a failing test to passing or vice versa. \emph{Weak Killing}~\cite{harman2011Strong} considers a mutant killed if there is any observable difference in execution behavior compared to the original program. This approach, exemplified by Metallaxis~\cite{papadakis2015metallaxis}, provides richer information for fault localization. Based on the killing analysis, for each mutant $m_i$, we define $T_k(m_i) = \{t \in T : \text{killed}(m_i, t) = \text{true}\}$ as the set of tests that kill the mutant, and $T_n(m_i) = T \setminus T_k(m_i)$ as the set of tests that do not kill the mutant.

\textbf{Step 4: Suspiciousness Calculation.} For each mutant $m_i \in \mathcal{M}(P)$, four key parameters are computed based on the intersection of mutant killing categories and test outcomes: $a_{np}(m_i) = |T_n(m_i) \cap T_p|$ (tests that do not kill the mutant and pass), $a_{kp}(m_i) = |T_k(m_i) \cap T_p|$ (tests that kill the mutant and pass), $a_{nf}(m_i) = |T_n(m_i) \cap T_f|$ (tests that do not kill the mutant and fail), and $a_{kf}(m_i) = |T_k(m_i) \cap T_f|$ (tests that kill the mutant and fail). These parameters are then input to a suspiciousness function $\mathcal{F}_{sus}$ to calculate the suspiciousness score $\emph{Sus}(m_i) = \mathcal{F}_{sus}(a_{np}(m_i), a_{kp}(m_i), a_{nf}(m_i), a_{kf}(m_i))$ for each mutant. Table~\ref{tab: suspicious formulas for MBFL} presents the most commonly used formulas of $\mathcal{F}_{sus}$ (i.e., the $\emph{Sus}(m)$ in the Formula column).

\begin{table}[hbtp]
\centering \caption{Suspiciousness Formulas for MBFL} 
\label{tab: suspicious formulas for MBFL}
{
\renewcommand{\arraystretch}{1.5}
\begin{tabular}{cc}
\toprule
\makecell[c]{\textbf{Name}} & \makecell{\textbf{Formula}} \\ \hline
Ochiai~\cite{abreu2006evaluation}  &   $\emph{Sus}(m)=\frac{a_{{kf}}(m)}{\sqrt{\left(a_{k f}(m)+a_{n f}(m)\right)\left(a_{k f}(m)+a_{k p}(m)\right)}}$      \\ 
Tarantula~\cite{jones2005empirical}& 
        $\emph{Sus}(m)=\frac{\frac{a_{kf}(m)}{a_{kf}(m)+a_{kp}(m)}}{\frac{a_{kf}(m)}{a_{kf}(m)+a_{nf}(m)}+\frac{a_{kp}(m)}{a_{kp}(m)+a_{np}(m)}}$
        \\ 
        
D$^*$~\cite{wong2013dstar}     &   $\emph{Sus}(m)=\frac{a_{k f}^*(m)}{a_{k p}(m) + a_{n f}(m)}$    \\ 
Jaccard~\cite{wong2016survey} & 
        
        $\emph{Sus}(m)=\frac{a_{kf}(m)}{a_{kf}(m)+a_{nf}(m)+a_{np}(m)}$
        
        \\
\bottomrule
\end{tabular}
}
\end{table}

\textbf{Step 5: Suspiciousness Aggregation.} After computing individual mutant suspiciousness scores, MBFL aggregates these scores to generate suspiciousness rankings for program entities $\mathcal{E}$. For each program entity $e \in \mathcal{E}$, the set of associated mutants is defined as $\mathcal{M}_e = \{m \in \mathcal{M}(P) : \text{loc}(m) = e\}$, where $\text{loc}(m)$ represents the program location where mutant $m$ was introduced. The suspiciousness scores are combined using an aggregation function $\mathcal{A}$, such that $\emph{Sus}(e) = \mathcal{A}(\{Sus(m) : m \in \mathcal{M}_e\})$, where $\mathcal{A}$ can be averaging, maximum, or weighted combinations. The final output is a ranked list $\mathcal{R} = \langle e_1, e_2, \ldots, e_{|\mathcal{E}|}\rangle$ where $\emph{Sus}(e_i) \geq Sus(e_{i+1})$.

\subsection{Mutant Interference}

Mutant interference represents a fundamental challenge that affects the reliability and accuracy of mutation-related computations across different application domains. Different types of mutants can introduce various forms of interference that impact the effectiveness of mutation testing, with distinct manifestations in test suite quality assessment and fault localization.

\subsubsection{Mutant Interference in Mutation Testing} 

Extensive research has investigated various types of mutants that interfere with mutation testing's primary objective of test suite quality assessment. \emph{Equivalent mutants}~\cite{tian2024large, schuler2013covering} are syntactically distinct from the original program but semantically identical, meaning no test case can ever ``kill'' them~\cite{tian2024large}. Their presence deflates the mutation score (i.e., the ratio of killed mutants to the total number of mutants), which can incorrectly suggest deficiencies in a test suite that do not exist~\cite{tian2024large}, leading to inaccurate assessments of test suite quality. \emph{Subsuming mutants} represent another category where mutant $m_{1}$ subsumes mutant $m_{2}$ if every test case detecting $m_{1}$ also detects $m_{2}$~\cite{garg2022cerebro}. When calculating mutation scores, including numerous redundant mutants artificially inflates the test suite's fault detection capability assessment. For instance, a test case might kill 10 mutually redundant mutants, but this does not indicate the test case's ability to detect 10 different types of faults. 
This leads to overly optimistic test quality assessments and can be misleading.

\subsubsection{Mutant Interference in Fault Localization} 

Due to the fundamental difference between test suite quality assessment and fault localization objectives, interference mutants in mutation testing typically have minimal impact on fault localization effectiveness. For MBFL, equivalent mutants have no impact on fault localization since they cannot be killed by test cases and thus provide no useful information for fault localization~\cite{papadakis2015metallaxis}. Similarly, subsuming mutants have limited impact on fault localization since the redundancy primarily affects test suite evaluation metrics rather than suspiciousness formula design assumption of MBFL.

However, fault localization faces a distinct class of interference mutants that create misleading diagnostic information by exhibiting fault-like killing behaviors without real fault relevance. Kim et al.~\cite{kim2023learning} analyzed the relationship between test cases and mutants, demonstrating that removing mutants with low coupling to faults could increase the number of faults localized at the first rank, providing concrete evidence that high-quality, denoised mutant subsets are more conducive to accurate fault localization. However, their approach relies primarily on statistical coupling analysis rather than semantic understanding of mutant behavior. Jang et al.~\cite{jang2022hotfuz} identified noise inherent to Higher-Order Mutants (HOMs), where suspiciousness scores are propagated to all associated statements, potentially causing non-faulty statements to receive high suspiciousness scores when coupled with faulty ones in an HOM. They proposed an HOM generation strategy to evenly distribute couplings and minimize negative impact on localization accuracy, though this approach focuses specifically on HOM-related noise rather than general interference from mutant analysis. Liu et al.~\cite{liu2024delta4ms, du2022improving} introduced the concept of \emph{mutant bias}, describing how differences between mutants from non-faulty and faulty code can distort fault localization accuracy by inflating scores of non-faulty statements or deflating those of faulty ones. They modeled this using signal theory and developed Delta4Ms to mitigate interference by removing false signals, but their approach treats interference as a statistical bias rather than addressing the underlying causes.

These findings reveal a critical bottleneck in current MBFL methods. Intuitively, mutants located near real faults are expected to be killed by more failing test cases and thus exhibit higher suspiciousness. However, previous studies~\cite{kim2023learning, jang2022hotfuz, liu2024delta4ms, du2022improving} reveal that many mutants derived from non-faulty code are also killed by a large number of test cases, creating interference in the fault localization process. These phenomena suggest that the relationship between mutant behavior and fault location is more complex than initially assumed.

\subsubsection{Limitations and Research Gaps} 

While existing studies provide valuable observations, they remain preliminary in terms of formal characterization and systematic approaches for recognition and mitigation. 
Current approaches primarily focus on statistical or heuristic methods to recognize and remove problematic mutants without a theoretical framework for understanding why certain mutants interfere with fault localization and how to recognize and mitigate them effectively. 
This gap motivates us to formally define and characterize the concept of FLIMs with theoretical analysis, and develop approaches to recognize and mitigate such mutants, thereby improving the accuracy and robustness of MBFL techniques.

\subsection{The Application of LLMs in Code Semantic Analysis}

The evolution of LLMs for code understanding represents a fundamental paradigm shift from traditional rule-based approaches to sophisticated neural architectures capable of advanced semantic reasoning. Early foundational works like Code2Vec~\cite{alon2019code2vec} and Code2Seq~\cite{alon2019code2seq} established the importance of incorporating program structure into machine learning models. 
The introduction of transformer-based architectures marked a critical inflection point, with CodeBERT~\cite{DBLP:journals/corr/abs-2002-08155} pioneering multi-modal pretraining and GraphCodeBERT~\cite{guo2021graphcodebert} advancing semantic understanding through data flow integration.
Contemporary LLMs have evolved along two primary trajectories: fine-tuning general-purpose foundation models on code-specific tasks (e.g., Code Llama~\cite{rozière2024codellama}, CodeFuse~\cite{10.1145/3639477.3639719}, Qwen2.5-Coder~\cite{yang2024qwen25coder}), and integrating code understanding into unified architectures for both natural language and code comprehension (e.g., GPT-4~\cite{openai2023gpt4}, Claude~\cite{anthropic2024claude35sonnet}, DeepSeek series~\cite{liu2024deepseekv3,guo2025deepseekr1}).

Recent advances in LLMs have demonstrated sophisticated capabilities for understanding complex program behavior beyond traditional syntactic structure~\cite{jelodar2025llm,zhao2024unveiling}, including deep understanding of functional semantics, data dependencies, and control flow logic.
This advanced semantic reasoning enables precise distinction between syntactically different but functionally equivalent code mutants, a fundamental requirement for mutation analysis and fault localization. Recent research has explored LLMs' capabilities in equivalent mutant detection, with studies like Empica~\cite{NGUYEN2025107780} revealing that while LLMs demonstrate capability in semantic analysis, there remains room for improvement in their sensitivity to semantically meaningful changes.

Despite these advances, two critical challenges affect the practical deployment of LLMs in software engineering tasks. First, \emph{domain adaptation challenges} arise from the gap between general-purpose LLM training and specialized code analysis requirements. Pre-trained models often lack sufficient sensitivity to subtle semantic differences in domain-specific contexts, leading to suboptimal performance in tasks like equivalent mutant detection where nuanced code understanding is crucial~\cite{jelodar2025llm}. Second, \emph{output reliability challenges} stem from the inherent instability and uncertainty in LLM outputs~\cite{xie2025empirical}, caused by stochastic sampling methods and probabilistic transformer architectures~\cite{bouchard2025uncertainty}. This manifests as inconsistent responses to identical inputs, even with deterministic settings, posing significant challenges for automated fault localization workflows requiring reliable results~\cite{assessing2025reliability}.

To address these challenges, domain-specific fine-tuning techniques~\cite{Ding2023ParameterEfficient} have emerged as essential approaches for bridging the domain gap, allowing adaptation of pre-trained models through supervised learning to enhance sensitivity to specialized code semantic patterns. Additionally, confidence estimation mechanisms~\cite{bouchard2025uncertainty} have become crucial for quantifying output reliability, providing independent algorithms to measure decision uncertainty and improve the robustness of LLM-based software engineering tools.

\section{Motivation Example}
\label{sec: Motivation}
In MBFL, the quality of mutants directly affects the accuracy of localization results. It is possible that mutants could be killed by failing test cases; however, the code entities of these mutants may not always align with the actual faulty locations. 
Such mutants often appear in code entities with complex interdependencies, and their ``fault-like'' behavior can mislead suspiciousness-based metrics, resulting in the inflation of suspiciousness scores for non-faulty code entities and consequently concealing the actual faulty code entities. These mutants are referred to as FLIMs.
For instance, in statement-level fault localization, FLIMs may emerge in statements with intricate control flow or data dependencies, leading to mislocalization of faults. To illustrate these phenomena, a statement-level example is presented as follows.

\begin{table*}[htb]
\centering
\caption{An Example of FLIMs Impact on Fault Localization}

\renewcommand{\arraystretch}{1}

\begin{tabular}{ccC{1.2cm}C{1cm}C{1cm}C{1cm}C{1cm}C{1cm}C{1cm}C{1cm}C{1cm}}
\toprule

\multicolumn{2}{c}{\textbf{Statement}}  & $S_{1}$ & \multicolumn{2}{c}{$S_{2}$ \textcolor{red}{\boldsymbol{$\times$}}} & $S_{3}$ \textcolor{red}{\boldsymbol{$\times$}} & \multicolumn{2}{c}{$S_{4}$} & \multicolumn{2}{c}{$S_{5}$} & $S_{6}$ \\ \midrule

\multicolumn{2}{c}{\textbf{Mutant}}  
& $m_{1}$ \textcolor{blue}{\boldsymbol{$\dagger$}} & $m_{2}$  &  $m_{3}$ &  $m_{4}$  &  $m_{5}$ \textcolor{blue}{\boldsymbol{$\dagger$}} & $m_{6}$ \textcolor{blue}{\boldsymbol{$\dagger$}}  & $m_{7}$ \textcolor{blue}{\boldsymbol{$\dagger$}} & $m_{8}$ \textcolor{blue}{\boldsymbol{$\dagger$}}  & $m_{9}$ \textcolor{blue}{\boldsymbol{$\dagger$}} \\ \midrule

\multirow{5}[6]{*}{\rotatebox[origin=c]{0}{\tabincell{c}{\textbf{Kill}\\\textbf{Matrix}}}} & $T_{p0}$      & -              & -              & -        & *              & -        & -              & -        & *              & -        \\\cmidrule(lr){2-2}\cmidrule(lr){3-11}
            & $T_{f0}$      & *              & *              & *        & -              & -        & *              & *        & *              & *        \\\cmidrule(lr){2-2}\cmidrule(lr){3-11}
            & $T_{p2}$      & -              & -              & *        & -              & -        & -              & *        & -              & *        \\\cmidrule(lr){2-2}\cmidrule(lr){3-11}
            & $T_{p3}$      & -              & -              & *        & -              & -        & -              & *        & -              & -        \\\cmidrule(lr){2-2}\cmidrule(lr){3-11}
            & $T_{f4}$      & *              & -              & -        & *              & *        & *              & *        & *              & -        \\ \midrule
\multirow{3}[3]{*}{{\rotatebox[origin=c]{0}{\tabincell{c}{\textbf{MBFL}}}}}        & $\emph{Sus}(M)$        & 1              & 0.71          & 0.41    & 0.5            & 0.71    & 1              & 0.71    & 0.82          & 0.71    \\\cmidrule(lr){2-2}\cmidrule(lr){3-11}
            & $\emph{Sus}(S)$        & 1 & \multicolumn{2}{c}{0.71} & 0.5 & \multicolumn{2}{c}{1} & \multicolumn{2}{c}{0.82} & 0.71 \\\cmidrule(lr){2-2}\cmidrule(lr){3-11}
            & $\emph{Rank}(S) $      & 1.5 & \multicolumn{2}{c}{4.5} & 6 & \multicolumn{2}{c}{1.5} & \multicolumn{2}{c}{3} & 4.5 \\\midrule
\multirow{3}[2]{*}{\rotatebox[origin=c]{0}{\tabincell{c}{\textbf{MBFL }\\\textbf{-}\\\textbf{FLIM}}}}   & $\emph{Sus}(M)$        & 0 & 0.71 & 0.41 & 0.5 & 0.71 & 0 & 0 & 0 & 0    \\\cmidrule(lr){2-2}\cmidrule(lr){3-11}
            & $\emph{Sus}(S)  $      & 0 & \multicolumn{2}{c}{0.71} & 0.5 & \multicolumn{2}{c}{0} & \multicolumn{2}{c}{0} & 0 \\\cmidrule(lr){2-2}\cmidrule(lr){3-11}
            & $\emph{Rank}(S)$      & 4.5 & \multicolumn{2}{c}{1} & 2 & \multicolumn{2}{c}{4.5} & \multicolumn{2}{c}{4.5} & 4.5 \\
\bottomrule

\end{tabular}

\footnotesize{\textit{Note: The \textcolor{red}{\boldsymbol{$\times$}} labeled statements are faulty statements; The \textcolor{blue}{\boldsymbol{$\dagger$}} labeled mutants are FLIMs; The } $\emph{Sus}(M)$ \textit{of MBFL are calculated using the Ochiai formula. The} $\emph{Rank}(M)$ \textit{are calculated using average ranking for tied scores.}}
\label{tab:flims-example}
\end{table*}

Table~\ref{tab:flims-example} illustrates a sample program composed of six statements and nine mutants with their fault localization data, demonstrating the influence of FLIMs in MBFL. The table uses a horizontal layout where the first row lists the statements ($S_{1}$ to $S_{6}$) and marks which statements are actual faulty statements ($S_{2}$ and $S_{3}$ marked with \textcolor{red}{\boldsymbol{$\times$}}). The second row lists the mutants (e.g., $m_{1}$ to $m_{9}$) and labels the ground truth FLIMs (marked with \textcolor{blue}{\boldsymbol{$\dagger$}}). The kill matrix part shows the execution results of each mutant under different test cases (e.g., $T_{f0}$, $T_{p2}$), where test cases with subscript ``f'' represent failing test cases and those with subscript ``p'' represent passing test cases, and ``*'' and ``-'' denote whether the mutant is killed by that test case.

In the MBFL part, the ``$\emph{Sus}(M)$'' and ``$\emph{Sus}(S)$'' represent the MBFL suspiciousness scores at the mutant and statement levels, respectively, and ``$\emph{Rank}(S)$'' gives the statement-level ranking. Although $S_{2}$ and $S_{3}$ are actual faulty statements, they do not rank at the top. This is due to the high suspiciousness scores assigned to FLIMs like $m_{1}$, $m_{5}$, $m_{6}$, $m_{7}$, $m_{8}$, and $m_{9}$, which belong to non-faulty statements (e.g., $S_{1}$, $S_{4}$, and $S_{5}$ have $\emph{Sus}(S)$ values of 1.0 or 0.82, ranking higher than the faulty statements).

To demonstrate the potential improvement achievable by addressing the FLIM problem, the MBFL-FLIM part shows the ideal scenario where all FLIMs are identified (marked with \textcolor{blue}{\boldsymbol{$\dagger$}}) and their suspiciousness scores are set to 0, i.e., the suspiciousness scores of FLIMs $m_{1}$, $m_{5}$, $m_{6}$, $m_{7}$, $m_{8}$, and $m_{9}$ are reduced to 0. The actual faulty statements $S_{2}$ and $S_{3}$ now achieve the top rankings (ranked 1 and 2 respectively), while the non-faulty statements with FLIMs drop to the bottom rankings (ranked 4.5, which is the average rank for the four statements $S_{1}$, $S_{4}$, $S_{5}$, and $S_{6}$ that all have suspiciousness scores of 0 and tie for positions 3-6).
This comparison clearly illustrates the significant negative impact of FLIMs on fault localization accuracy and demonstrates the substantial potential for improvement if FLIMs can be effectively identified and handled.

The demonstrated impact of FLIMs on fault localization effectiveness motivates a theoretical analysis. The following section provides a formal definition and theoretical foundation for FLIMs, establishes the underlying interference mechanisms, and analyzes FLIMs behavioral root causes using the RIPR model.

\section{FLIMs: Fault Localization Interference Mutants}
\label{sec: FLIMs}

\subsection{The Suspiciousness Formula Design Assumption of MBFL}

The effectiveness of MBFL fundamentally relies on the suspiciousness formula design assumption that mutants generated from faulty statements show different test sensitivity from those from non-faulty statements. Let program $P$ contain faulty statement $s_e$ and non-faulty statements $s_c \in S \setminus \{s_e\}$. Given test suite $T = T_f \cup T_p$ where $T_f$ fails on $P$ and $T_p$ passes, MBFL generates two mutant classes:

\begin{itemize}
    \item $M_e$: Mutants of $s_e$ (fault-relevant)
    \item $M_c$: Mutants of $s_c$ (non-fault-relevant)
\end{itemize}

This assumption asserts that mutants from faulty statements should exhibit significantly higher kill rates when exposed to failing tests compared to mutants from non-faulty statements~\cite{papadakis2015metallaxis, moon2014ask}:

\begin{equation}
    \frac{|\mathcal{K}(M_e, T_f)|}{|M_e|} \gg \frac{|\mathcal{K}(M_c, T_f)|}{|M_c|}
    \label{eq:core_principle}
\end{equation}

where $\mathcal{K}(M_x, T_f) \triangleq \{m \in M_x \mid \exists t_f \in T_f: \mathcal{K}(m, t_f)\}$ denotes mutants killed by failing tests, with $\mathcal{K}(m, t)$ indicating test $t$ detects mutant $m$.

The degree to which this assumption is satisfied in practice directly determines the effectiveness of MBFL. When the assumption holds strongly (i.e., in Equation~\ref{eq:core_principle} $\frac{|\mathcal{K}(M_e, T_f)|}{|M_e|}$ is significantly large than $\frac{|\mathcal{K}(M_c, T_f)|}{|M_c|}$), MBFL achieves high fault localization accuracy, enabling precise identification of faulty statements. Conversely, when this assumption is violated or weakened by interference factors, MBFL's fault localization effectiveness deteriorates, leading to reduced accuracy in localizing the actual fault entities.

\subsection{Definition of FLIMs}
\label{subsec:flim_formulation}

While the suspiciousness formula design assumption of MBFL described above generally holds in practice, it may be challenged by certain anomalous behaviors that can potentially weaken this assumption. There exists an important phenomenon: certain mutants of non-faulty code entities $m_c \in M_c$ violate this fundamental relationship by being killed by at least one failing test ($\exists t_f \in T_f: \mathcal{K}(m_c, t_f)$). We designate such mutants as Fault Localization Interference Mutants (FLIMs).

\begin{definition}[FLIMs]

\label{def:flim}
A mutant of a non-faulty code entity $m_c \in M_c$ is classified as a FLIM if and only if it is killed by at least one failing test:
\begin{equation}
    \text{FLIMs}(m_c) \iff \exists t_f \in T_f: \mathcal{K}(m_c, t_f)
    \label{eq:flim_def}
\end{equation}

\end{definition}

This definition is applicable irrespective of the number of faults present in the program, including single and multiple faults, and accounts for varying numbers of failing tests.
Note that since $m_c$ originates from non-faulty statements in $M_c$, it inherently cannot repair the original fault in $P$, making the non-repair condition implicit in our definition. FLIMs create false diagnostic information by mimicking the test sensitivity behaviors of real faulty statements, thereby weakening the suspiciousness formula design assumption of MBFL. Specifically, when FLIMs are killed by failing tests, they artificially inflate the suspiciousness scores of non-faulty statements, potentially causing these statements to rank higher than the actual faulty statements in the final ranking. This directly undermines the fundamental assumption in Equation~\ref{eq:core_principle} that fault-relevant mutants should exhibit significantly higher kill rates with failing tests compared to non-fault-relevant ones.

\begin{corollary}[FLIMs Location Property]
\label{corollary:flim_location}
For any FLIM $m_c$ that is killed by a failing test $t_f$, the mutation location of $m_c$ must reside on at least one error propagation path from the original fault to the observable failure in $t_f$.
\end{corollary}

\begin{proof}
Let $s_f$ denote the original faulty statement and $s_c$ denote the non-faulty statement where FLIM $m_c$ is generated. For $m_c$ to be killed by failing test $t_f$ (i.e., $\mathcal{K}(m_c, t_f)$), the mutation must create observable behavioral differences when $t_f$ executes the mutated program compared to the original faulty program.

Since $t_f$ fails on the original program $P$, there exists an error propagation path $\pi: s_f \rightsquigarrow \text{output}$ where the fault at $s_f$ propagates through a sequence of program statements to produce the observable failure. For the mutation at $s_c$ to affect the test outcome, it must interact with this error propagation process.

We prove by contradiction. Assume $s_c \notin \pi$ for all error propagation paths $\pi$ from $s_f$ to the failure output. This means $s_c$ is not involved in any computation that processes, transforms, or transmits the error state originated from $s_f$. In this case:

\begin{enumerate}
    \item The error state $E_f$ produced by $s_f$ propagates to the output independently of any computation at $s_c$
    \item The mutation at $s_c$ produces a local state change $E_m$ that cannot interact with $E_f$ during program execution
    \item The final program output remains determined solely by the error propagation from $s_f$, unaffected by the mutation at $s_c$
\end{enumerate}

Therefore, $t_f$ would produce identical failure behaviors on both the original program and the mutated program, contradicting the assumption that $\mathcal{K}(m_c, t_f)$. Hence, $s_c$ must reside on at least one error propagation path from $s_f$ to the observable failure.
\end{proof}

This corollary reveals a fundamental constraint on FLIMs locations: they cannot arbitrarily appear at any non-faulty statement, but are inherently limited to those statements that participate in error propagation processes. This insight has important implications for FLIMs recognition, as it suggests that analyzing the dynamic error propagation paths can help identify potential FLIM locations and distinguish them from fault-irrelevant mutants.

\subsection{FLIM Interference Analysis}
\label{subsec:flim_analysis}

\subsubsection{RIPR Model: From Testing to Fault Localization}

To understand FLIMs interference root causes, we first introduce the Reachability, Infection, Propagation, Revealability (RIPR) model~\cite{li2016test}, which provides a comprehensive framework for analyzing how faults manifest as observable failures. The RIPR model, evolved from the earlier PIE model~\cite{voas1992pie}, identifies four necessary conditions for a test to detect a fault: the test execution must reach the faulty code location (Reachability), the fault must create distinguishable erroneous program states (Infection), the infected states must propagate to observable program outputs (Propagation), and the test oracle must detect and report the propagated errors (Revealability).

While RIPR has been extensively applied in mutation testing for test case design~\cite{li2016test, hierons2009mutation, fraser2012mutation}, in this paper, we introduce its application to MBFL for analyzing interference root causes. To distinguish between these two application contexts, we denote the RIPR model in the software testing field as RIPR$_{ST}$ and our fault localization-adapted RIPR model as RIPR$_{FL}$.

In the RIPR$_{FL}$ model, we focus on understanding how mutants interact with existing program faults during test execution. Specifically, a mutant $m$ satisfies the RIPR$_{FL}$ conditions with respect to test $t$ when: (1) \textsc{Reach}$_{FL}(m,t)$: test $t$ executes the mutant's modified code location, (2) \textsc{Infect}$_{FL}(m)$: the mutant creates distinguishable program state differences compared to the original program execution, (3) \textsc{Propagate}$_{FL}(m)$: the infected state differences propagate to observable program outputs or behaviors, and (4) \textsc{Reveal}$_{FL}(m,t)$: test $t$ detects the propagated differences through either strong killing (pass/fail changes) or weak killing (observable behavioral changes). This can be formally expressed as $\mathcal{K}_{FL}(m,t) \iff \textsc{Reach}_{FL}(m,t) \land \textsc{Infect}_{FL}(m) \land \textsc{Propagate}_{FL}(m) \land \textsc{Reveal}_{FL}(m,t)$.

This adaptation requires clarifying the fundamental differences in how RIPR applies to mutation testing versus fault localization contexts. In mutation testing, RIPR is used \emph{prescriptively} as a design framework. Test engineers use it to construct test cases that can successfully traverse all four phases to kill mutants. The goal is \emph{forward causation}: designing inputs that can trigger a complete RIPR chain from a potential fault to an observable failure.

In contrast, in this paper, we apply RIPR \emph{analytically} as a diagnostic framework in the fault localization context. We start with an already confirmed failure (revealability condition satisfied) and work backwards through the propagation traces to identify the original infection source. The goal is \emph{reverse causation}: analyzing existing execution data to infer the most likely fault location. The model serves as an explanatory lens for understanding why certain program elements appear suspicious.

This contextual shift fundamentally changes the interpretation of each RIPR phase. In fault localization, we assume that failing tests have already established complete RIPR chains, and our task is to distinguish between elements that serve as infection sources (i.e., actual faults) versus those that merely participate in propagation pathways (i.e., non-faulty statements in the error transmission route).

\subsubsection{FLIM Interference Root Causes Analysis}

FLIMs represent a problematic class of mutants that satisfy the RIPR$_{FL}$ conditions for failing tests while preserving the original program fault. Their interference occurs through systematic exploitation of the complex interactions between the original fault and the introduced mutation during the RIPR$_{FL}$ phases. We identify the primary interference root causes based on where and how these interactions manifest:

\textbf{Reach-based Interference} occurs when FLIMs alter the program's control flow in a way that changes the execution path taken by failing tests. The mutant may introduce new conditional branches, modify loop conditions, or alter function call sequences. When a failing test $t_f$ executes the mutated program, it may follow a different execution path that either bypasses the original fault entirely or reaches it under different program states. If this alternative path leads to correct output (strong killing: $f \to p$) or different incorrect output (weak killing: $f \to f'$), the mutant appears to ``fix" or ``change" the failure, creating a false indication that the mutated location is fault-relevant.

\textbf{Infection-based Interference} operates at the program state level, where FLIMs introduce state changes that interact with the error states produced by the original fault. When both the original fault and the mutant are executed by the same failing test, they create a composite error state. The key insight is that these two error sources can exhibit \emph{error masking} behavior - the state changes introduced by the mutant can neutralize, overwrite, or redirect the error states caused by the original fault. This masking effect can occur through direct variable overwrites, computational cancellations (e.g., $+1$ and $-1$ effects), or logical contradictions that resolve to correct values.

\textbf{Propagation-based Interference} targets the error transmission pathways through which infected states flow to program outputs. FLIMs can alter how errors propagate by modifying intermediate computations, changing data flow behaviors, or redirecting information through different program variables or structures. The mutant essentially creates alternative propagation channels that can either block the original error from reaching outputs (resulting in correct final states) or route it through different computational paths that transform the error signature. When the propagation interference successfully alters the final program behavior, it can lead to strong killing ($f \to p$) if the error is completely neutralized, or weak killing ($f \to f'$) if the error signature is transformed but still detectable by the test oracle.

\textbf{Revealability-based Interference} exploits weaknesses or gaps in test oracles and assertion coverage. FLIMs can change program behavior in ways that satisfy the existing test assertions while preserving the underlying logical fault. This occurs when the mutant alters program outputs to values that happen to pass the current test checks, even though the computation leading to these outputs remains fundamentally flawed. The interference leverages the fact that test oracles often check only specific aspects of program behavior, leaving other aspects unverified.

\subsubsection{FLIM Interference Behavior Analysis}

In practice, FLIMs exhibit disruptive behavior. Although four causes of interference were identified in the previous subsection, FLIMs typically combine multiple causes simultaneously in real-world scenarios, creating complex interference patterns that are particularly challenging for MBFL techniques to handle.

\textbf{Single Cause Scenarios:}
When FLIMs operate through a single interference root cause, their behavior is relatively straightforward to understand. For example, a reach-based interference mutant may simply redirect program execution to bypass the original fault, leading to correct output for failing tests. Similarly, a revealability-based interference mutant may alter program outputs just enough to satisfy test assertions while preserving the underlying fault.

\textbf{Multiple Causes Combinations:}
More problematic are FLIMs that combine multiple interference causes. For instance, a mutant may simultaneously exhibit infection-based and propagation-based interference: it introduces state changes that interact with the original fault's error states (infection-based), while also altering how these composite error states flow through the program (propagation-based). Such combined interference creates a cascading effect where the mutant's impact is amplified across multiple RIPR$_{FL}$ phases.

To illustrate this complexity, consider a failing test $t_f$ executing a mutant $m_c$ that combines infection and propagation interference. The test encounters two error sources: the original fault $s_f$ producing error state $E_f$, and the mutant $m_c$ producing state perturbation $E_m$. During execution, these error states interact in ways that can neutralize each other. The mutant's perturbation $E_m$ may systematically counteract the fault's error state $E_f$, leading to correct final output despite both error sources being active.

This is referred to as \emph{coincidental correctness}, whereby the program reaches correct answers through incorrect computations. Specifically, the incorrect computations arise from the superposition of two types of errors: the mutant-introduced faults and the program's original faults. When these compound errors interact during program execution, they can accidentally cancel each other out, producing correct results despite the presence of multiple faults. For MBFL techniques like MUSE that rely on strong killing signals ($f \to p$ transitions) as evidence of fault repair, such coincidental correctness represents a fundamental challenge. The technique incorrectly interprets the mutant location as fault-relevant when it actually introduces additional faults.

\textbf{Challenges for MBFL:}
The complex behavioral patterns of FLIMs pose significant challenges for current MBFL approaches, fundamentally undermining their effectiveness by creating misleading diagnostic indicators that traditional fault localization techniques struggle to distinguish from real fault indicators. The complexity of FLIMs interference, particularly when multiple causes combine, makes them difficult to detect and mitigate using conventional dynamic analysis approaches, which rely primarily on execution monitoring and provide limited semantic understanding of why certain mutants exhibit interfering behavior. This limitation motivates the need for more sophisticated analysis techniques that can leverage semantic understanding of program logic and fault patterns.

\section{FLIM Recognition and Mitigation}
\label{sec:flims-recognition-mitigation}

To alleviate the adverse impact of FLIMs on MBFL effectiveness, this section presents our approach for FLIM recognition and mitigation. (1) LLM-based FLIM Recognition, which utilizes the semantic understanding capabilities of LLMs to identify FLIMs; (2) FLIM Mitigation, which refines the suspiciousness values of mutants based on the identified FLIMs to improve fault localization effectiveness.

\subsection{LLM-based FLIM Recognition}
\label{subsec:llm-flims-recognition}

The recognition method is based on LLM and comprises three components that work synergistically: (1) \emph{Recognition Prompt Engineering}, which designs structured prompts to guide LLMs in performing accurate FLIM classification; (2) \emph{Fine-tuning}, an optional enhancement that adapts LLMs to domain-specific FLIM characteristics through supervised learning; and (3) \emph{Confidence Estimation}, a further optional component that provides independent confidence estimation algorithms for FLIM recognition results.

\subsubsection{Prompt Engineering}
\label{subsubsec:recognition-prompt-engineering}

Prompt Engineering forms the foundation of our LLM-based approach, focusing on the design and construction of structured prompts that guide LLMs to accurately distinguish FLIMs from fault-revealing mutants. This component operates by extracting features from test execution and mutant analysis, constructing structured prompts, and processing LLM outputs to generate FLIM recognition results.

\paragraph{Feature Extraction from Test Execution and Mutant Analysis}

To enable effective FLIM recognition, we extract both static and dynamic analysis features from the MBFL process. Let $\mathcal{A}(m_i, T)$ denote the feature extraction function that produces a feature vector for mutant $m_i$ given test suite $T$. The extracted features include:

\begin{itemize}
\item \emph{Static Features}: For each mutant $m_i$, we extract the mutation location $\text{loc}(m_i)$, the original code snippet $\text{code}_{\text{orig}}(m_i)$, and the mutated code snippet $\text{code}_{\text{mut}}(m_i)$. The code difference is represented as $\Delta(m_i) = \text{diff}(\text{code}_{\text{orig}}(m_i), \text{code}_{\text{mut}}(m_i))$.

\item \emph{Dynamic Features}: From test execution, we collect the original program's failing test information $\mathcal{O}_{\text{fail}}(P, T_f) = \{(t, \text{error}(t), \text{trace}(t)) : t \in T_f\}$, where $\text{error}(t)$ and $\text{trace}(t)$ represent the error message and stack trace for failing test $t$. Similarly, for mutant $m_i$, we obtain $\mathcal{O}_{\text{fail}}(m_i, T_f)$.

\item \emph{Behavioral Differences}: We compute the differences in test execution behaviors between the original program and mutant, including error message differences $\Delta_{\text{error}}(m_i, t) = \text{diff}(\text{error}_P(t), \text{error}_{m_i}(t))$ and stack trace differences $\Delta_{\text{trace}}(m_i, t) = \text{diff}(\text{trace}_P(t), \text{trace}_{m_i}(t))$ for each failing test $t$.
\end{itemize}

\paragraph{Prompt Design and Construction}

The extracted features are organized into a structured prompt $\mathcal{P}(m_i)$ that guides the LLM to perform FLIM recognition. Our prompt design follows a systematic approach that includes role definition, task specification, feature presentation, and output formatting. The prompt structure is designed to leverage the LLM's semantic understanding while providing sufficient context for accurate FLIM classification.

The prompt construction function can be formalized as:
\begin{equation}
\mathcal{P}(m_i) = \text{construct}(\mathcal{A}(m_i, T), \text{template})
\end{equation}

where $\text{template}$ defines the structured format for presenting mutant information and task requirements to the LLM. Figure~\ref{fig:prompt-template} shows the details of our designed FLIM recognition prompt template.

\begin{figure}[htb]
\centering
\includegraphics[width=0.48\textwidth]{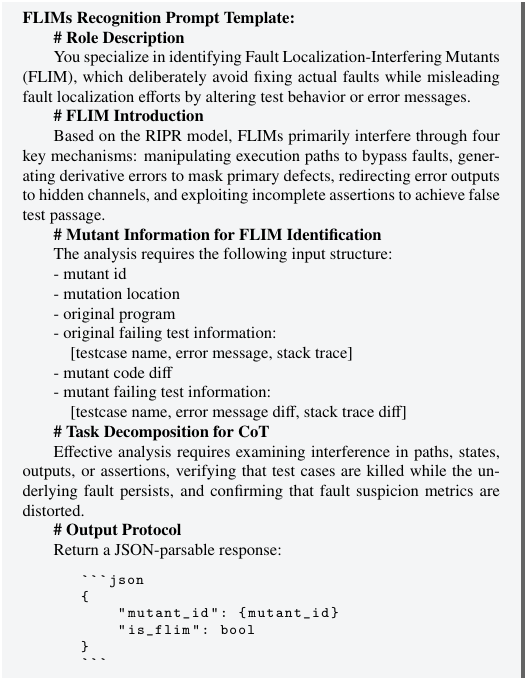}
\caption{FLIM Recognition Prompt Template}
\label{fig:prompt-template}
\end{figure}

\paragraph{LLM Processing and Result Generation}

Given the constructed prompt $\mathcal{P}(m_i)$, the LLM processes the information and generates a response $\mathcal{R}(m_i)$. The response is then parsed to extract the FLIM recognition result. Let $\psi: \mathcal{R}(m_i) \rightarrow \{0, 1\}$ denote the parsing function that converts the LLM output into a binary classification result, where 1 indicates that the mutant is identified as a FLIM and 0 indicates otherwise.

The complete LLM-based recognition process can be expressed as:
\begin{equation}
\text{result}(m_i) = \psi(\text{LLM}(\mathcal{P}(m_i)))
\end{equation}

The recognition results for all mutants form a binary vector $\mathbf{r} = [\text{result}(m_1), \text{result}(m_2), \ldots, \text{result}(m_{|\mathcal{M}|})]^T \in \{0,1\}^{|\mathcal{M}|}$, which serves as the foundation for subsequent confidence estimation and suspiciousness refinement.

\subsubsection{Fine-tuning}
\label{subsubsec:fine-tuning}

The Fine-tuning component is an optional enhancement designed to address the insufficient sensitivity of general-purpose LLMs to subtle code semantic changes that may affect FLIM recognition accuracy. This component adapts pre-trained language models to domain-specific FLIM characteristics through supervised learning, thereby improving recognition performance on software fault localization tasks.

\paragraph{Fine-tuning Dataset Construction}

The construction of the fine-tuning dataset is crucial for FLIM recognition. We partition the dataset into training and testing sets and use the training set to fine-tune the LLM. Let $\mathcal{D}_{\text{train}} = \{(m_i, y_i)\}_{i=1}^{N}$ denote the training dataset, where $m_i$ represents a mutant with its associated features and $y_i \in \{0, 1\}$ is the corresponding FLIM label. 

The ground truth labeling of FLIMs is performed using MBFL suspiciousness scores and fault location information. Specifically, for each program $P$ in the training set with known fault locations $\mathcal{F}$, we generate the complete mutant set $\mathcal{M}(P)$ and compute their suspiciousness scores using MBFL techniques. 
The FLIM labeling strategy is defined as follows: a mutant $m \in \mathcal{M}(P)$ is labeled as a FLIM ($y = 1$) if its suspiciousness score is non-zero and it is located on non-faulty program entities, otherwise it is labeled as non-FLIM ($y = 0$). 
This criterion labels mutants that receive suspiciousness scores from MBFL techniques but do not belong to the actual fault locations, which may exhibit interference that affects fault localization accuracy.

\paragraph{Fine-tuning Process}

The fine-tuning process employs Low-Rank Adaptation (LoRA) to efficiently adapt the pre-trained LLM while avoiding overfitting. The fine-tuning objective can be formalized as:

\begin{equation}
\theta^* = \arg\min_{\theta} \sum_{(m_i, y_i) \in \mathcal{D}_{\text{train}}} \mathcal{L}_{\text{CE}}(f_{\theta}(\mathcal{P}(m_i)), y_i) + \lambda \|\Delta\theta\|_2^2
\end{equation}

where $f_{\theta}$ represents the LLM with parameters $\theta$, $\mathcal{L}_{\text{CE}}$ is the cross-entropy loss function, $\Delta\theta$ denotes the parameter updates, and $\lambda$ is the regularization coefficient. The LoRA technique restricts updates to low-rank matrices, significantly reducing computational requirements while maintaining model performance.

\subsubsection{Confidence Estimation}
\label{subsubsec:confidence-estimation}

The Confidence Estimation component addresses the inherent instability of LLM outputs by providing independent confidence estimation algorithms for FLIM recognition results. This component is designed as a modular and ablatable component, where each confidence estimation algorithm operates independently to tackle the challenge that LLMs may produce inconsistent outputs when processing the same input multiple times.

\paragraph{Multi-Decision Matrix Construction}

To capture the uncertainty in LLM decisions, we perform multiple inference runs for each mutant. Let $\mathbf{R} \in \{0,1\}^{|\mathcal{M}| \times K}$ denote the multi-decision matrix, where $K$ represents the number of repeated inference runs. Each element $R_{i,j}$ represents the FLIM recognition result for mutant $m_i$ in the $j$-th inference run:

\begin{equation}
R_{i,j} = \psi(\text{LLM}_j(\mathcal{P}(m_i)))
\end{equation}

where $\text{LLM}_j$ denotes the $j$-th inference run and $\psi$ is the result parsing function defined earlier. The multi-decision matrix captures the variability in LLM outputs and serves as the foundation for confidence estimation.

\paragraph{Confidence Estimation Algorithms}

We provide three independent algorithms for computing confidence scores from the multi-decision matrix. Each algorithm represents a distinct approach to measuring decision reliability and can be used independently as an ablatable component.

\emph{Entropy-based Weighting (EBW):} This method leverages information entropy to measure the uncertainty across different inference runs. For each mutant $m_i$, we first compute the probability distribution over decisions:

\begin{equation}
p_i^{(0)} = \frac{1}{K}\sum_{j=1}^{K}(1-R_{i,j}), \quad p_i^{(1)} = \frac{1}{K}\sum_{j=1}^{K}R_{i,j}
\end{equation}

The entropy for mutant $m_i$ is then calculated as:

\begin{equation}
H_i = -p_i^{(0)}\log p_i^{(0)} - p_i^{(1)}\log p_i^{(1)}
\end{equation}

The confidence score is derived from the normalized entropy:

\begin{equation}
\phi_{\text{EBW}}(m_i) = 1 - \frac{H_i}{\log 2}
\end{equation}

where higher entropy indicates lower confidence and vice versa.

\emph{Principal Component Analysis (PCA):} This method reduces the dimensionality of the decision matrix while preserving the most informative variance. We apply PCA to the multi-decision matrix $\mathbf{R}$ to obtain the first principal component:

\begin{equation}
\mathbf{w} = \arg\max_{\|\mathbf{w}\|=1} \text{Var}(\mathbf{R}\mathbf{w})
\end{equation}

The confidence scores are computed by projecting each row onto the principal component and normalizing:

\begin{equation}
\phi_{\text{PCA}}(m_i) = \text{MinMaxScale}(\mathbf{R}_i \cdot \mathbf{w})
\end{equation}

where $\mathbf{R}_i$ represents the $i$-th row of the decision matrix and MinMaxScale normalizes the values to the $[0,1]$ interval.

\emph{Count-based Weighting (CBW):} This method directly uses the proportion of positive decisions as the confidence score:

\begin{equation}
\phi_{\text{CBW}}(m_i) = \frac{1}{K}\sum_{j=1}^{K}R_{i,j}
\end{equation}

This approach provides an intuitive measure of confidence based on the consistency of decisions across multiple runs.

Each of three algorithms can be selected independently based on experimental requirements, making the confidence estimation component highly flexible and suitable for ablation studies.

\subsection{FLIM Mitigation}
\label{subsec:flims-mitigation}

In order to mitigate the adverse impact of FLIMs on MBFL, it is necessary to refine the suspiciousness values of the FLIMs, i.e., FLIM Mitigation.

For each mutant $m_i \in \mathcal{M}(P)$, we refine its original suspiciousness score $\text{Sus}_{\text{orig}}(m_i)$ using a straightforward scaling approach based on the FLIM recognition results. The refinement strategy differs depending on whether confidence estimation is employed:

\emph{Basic Version (Binary FLIM Results):} When using only the LLM-based FLIM recognition without confidence estimation, the refinement is performed using binary results:

\begin{equation}
\text{Sus}_{\text{ref}}(m_i) = \text{Sus}_{\text{orig}}(m_i) \times (1 - \text{result}(m_i))
\end{equation}

where $\text{result}(m_i) \in \{0, 1\}$ is the binary FLIM recognition result. This approach completely eliminates the suspiciousness of mutants identified as FLIMs while preserving the original suspiciousness of non-FLIMs.

\emph{Confidence-enhanced Version:} When employing confidence estimation, the refinement incorporates the confidence scores:

\begin{equation}
\text{Sus}_{\text{ref}}(m_i) = \text{Sus}_{\text{orig}}(m_i) \times (1 - \phi(m_i))
\end{equation}

where $\phi(m_i) \in [0,1]$ represents the confidence score from any of the three independent confidence estimation algorithms. This approach provides a more nuanced refinement that gradually reduces suspiciousness based on the confidence level of FLIM recognition.

After refining mutant-level suspiciousness scores, we aggregate them to compute the final suspiciousness for each program entity $e \in \mathcal{E}$. Let $\mathcal{M}_e = \{m \in \mathcal{M}(P) : \text{loc}(m) = e\}$ denote the set of mutants located at program entity $e$. We employ the maximum aggregation strategy, which has proven effective in highlighting the most suspicious mutants:

\begin{equation}
\text{Sus}_{\text{final}}(e) = \max_{m \in \mathcal{M}_e} \text{Sus}_{\text{ref}}(m)
\end{equation}

This aggregation method ensures that program entities containing at least one highly suspicious (non-FLIMs) mutant receive high suspiciousness scores, while entities containing only FLIMs receive reduced scores.

Finally, program entities are ranked based on their refined suspiciousness scores to produce the fault localization result. This design supports ablation studies and allows researchers to assess the effectiveness of different FLIM recognition approaches independently.

\section{MBFL-FLIM}
\label{sec:mbfl-flim-framework}

This section presents the MBFL-FLIM framework, which integrates the LLM-based FLIM recognition and mitigation techniques into the MBFL framework.

\subsection{MBFL-FLIM Framework}
\label{subsec:framework-overview}

The MBFL-FLIM framework is designed to seamlessly integrate FLIM recognition and mitigation capabilities into existing MBFL workflows. The framework is illustrated in Figure~\ref{fig:Workflow of MBFL-FLIM}.

\begin{figure*}[htb]
    \centering
    \centering\includegraphics[width=18cm]{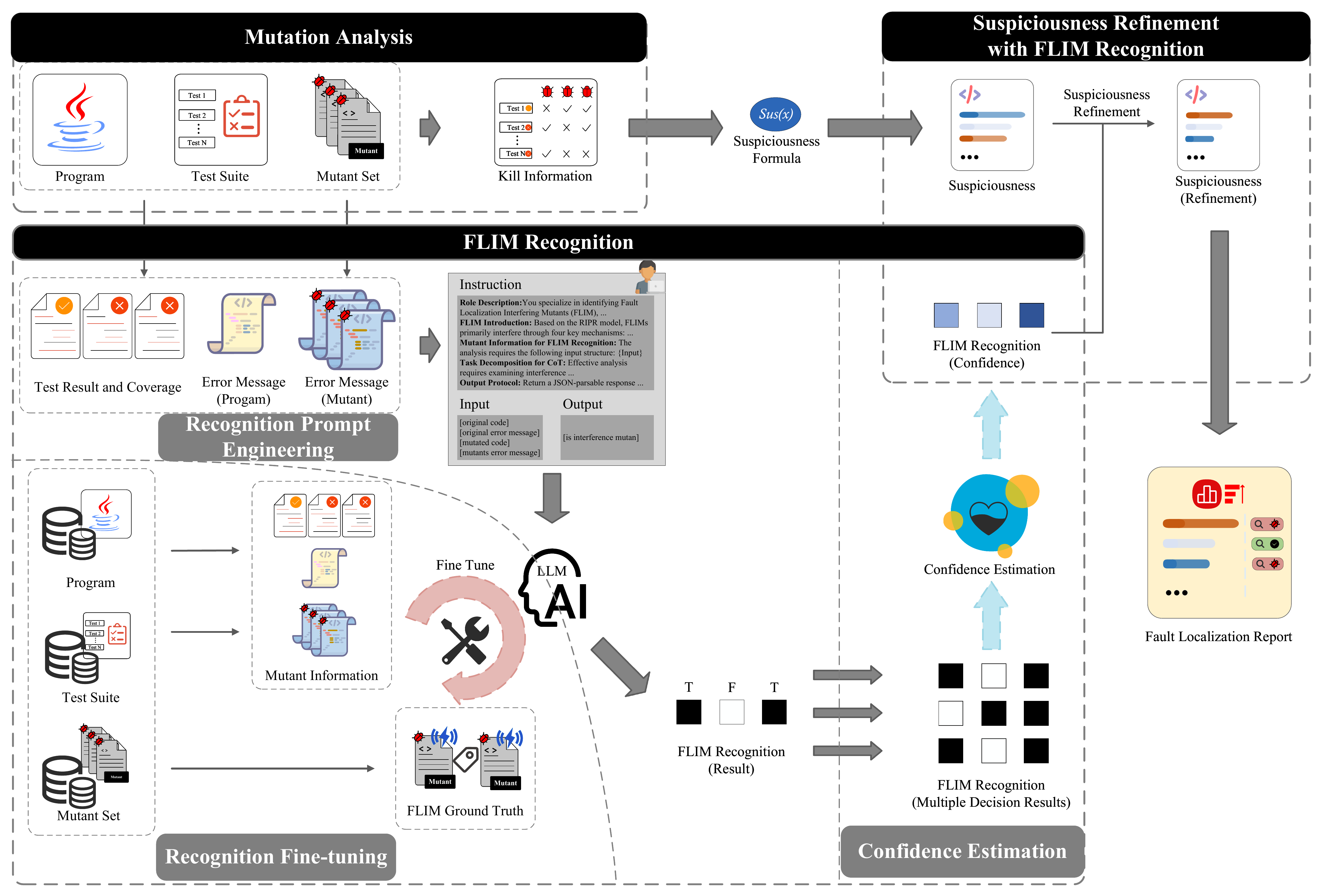}
    \caption{Workflow of MBFL-FLIM}
    \label{fig:Workflow of MBFL-FLIM}
\end{figure*}

The framework begins with standard MBFL test execution and mutant analysis to extract necessary code, mutant, and test features for FLIM recognition. The LLM-based FLIM recognition then processes these features through carefully designed prompts to distinguish FLIMs. The fine-tuning component enhances model performance using domain-specific training data labeled based on known MBFL results. The confidence estimation component processes multiple LLM decisions to generate stable confidence scores. Finally, the FLIM mitigation component refines mutant suspiciousness values using the FLIM recognition results and confidence scores to improve the fault localization effectiveness.

\subsection{MBFL-FLIM Algorithm}
\label{subsec:mbfl-flim-algorithm}

Algorithm~\ref{alg:MBFL-FLIM} formalizes the MBFL-FLIM framework by integrating the FLIM recognition and mitigation components into the standard MBFL procedure. The algorithm extends traditional MBFL with FLIM recognition capabilities and confidence-based suspiciousness refinement.

\begin{algorithm}[!htb]
\caption{The MBFL-FLIM Algorithm}
\label{alg:MBFL-FLIM}
\small
\begin{algorithmic}[1]
\allowdisplaybreaks
\Require{Program $P$, Test suite $T$, Mutation operators $\mathcal{O}$, LLM}
\Ensure{Ranked list $\mathcal{R}_{\text{MBFL-FLIM}}$ of program entities}

\State \textbf{Step 1: Test Execution and Classification}
\State Execute test suite $T$ on original program $P$
\State Record test outcomes: $\text{outcome}(P, T)$
\State Classify tests: $T_{\text{pass}} = \{t \in T : \text{outcome}(P, t) = \text{pass}\}$
\State $T_{\text{fail}} = \{t \in T : \text{outcome}(P, t) = \text{fail}\}$
\State Extract covered entities: $\mathcal{E} = \{e : e \in \text{Coverage}(P, T_{\text{fail}})\}$
\State \textbf{Step 2: Mutant Generation and Execution}
\State Generate mutant set: $\mathcal{M}(P) = \bigcup_{\mu \in \mathcal{O}} \mu(P, \mathcal{E})$
\For{each mutant $m_i \in \mathcal{M}(P)$}
    \State Execute $T$ on $m_i$ and record outcomes: $\text{outcome}(m_i, T)$
    \State Record mutant location: $\text{loc}(m_i)$
    \State Compute kill matrix: $\text{Kill}(m_i, t_j)$ for all $t_j \in T$
\EndFor
\State \textbf{Step 3: FLIM Recognition}
\State \textbf{\textit{Step 3.1: Recognition Prompt Engineering}}
\For{each mutant $m_i \in \mathcal{M}(P)$}
    \State Extract features: $\mathcal{A}(m_i, T)$
    \State $\mathcal{P}(m_i) = \text{construct}(\mathcal{A}(m_i, T), \text{template})$
    \State $\text{result}(m_i) = \psi(\text{LLM}(\mathcal{P}(m_i)))$
\EndFor
\State \textbf{\textit{Step 3.2: Fine-tuning (Optional)}}
\If{fine-tuning enabled}
    \For{each mutant $m_i \in \mathcal{M}(P)$}
        \State $\text{result}(m_i) = \psi(\text{LLM}_{\text{tuned}}(\mathcal{P}(m_i)))$
    \EndFor
\EndIf
\State \textbf{\textit{Step 3.3: Confidence Estimation (Optional)}}
\If{confidence estimation enabled}
    \For{each mutant $m_i \in \mathcal{M}(P)$}
        \State $\mathbf{R}_{i,j} = \psi(\text{LLM}_j(\mathcal{P}(m_i)))$ for $j = 1, \ldots, K$
        \State $\phi(m_i) = \text{ConfidenceAlgorithm}(\mathbf{R}_i)$
    \EndFor
\EndIf
\State \textbf{Step 4: MBFL Suspiciousness Computation}
\For{each mutant $m_i \in \mathcal{M}(P)$}
    \State Compute killing statistics: $a_{np}, a_{kp}, a_{nf}, a_{kf}$
    \State $\text{Sus}_{\text{orig}}(m_i) = \text{MBFL}(a_{np}, a_{kp}, a_{nf}, a_{kf})$
\EndFor
\State \textbf{Step 5: FLIM Mitigation}
\For{each mutant $m_i \in \mathcal{M}(P)$}
    \If{confidence estimation enabled}
        \State $\text{Sus}_{\text{ref}}(m_i) = \text{Sus}_{\text{orig}}(m_i) \times (1 - \phi(m_i))$
    \Else
        \State $\text{Sus}_{\text{ref}}(m_i) = \text{Sus}_{\text{orig}}(m_i) \times (1 - \text{result}(m_i))$
    \EndIf
\EndFor
\State \textbf{Step 6: Program Entity Aggregation and Ranking}
\For{each program entity $e \in \mathcal{E}$}
    \State $\mathcal{M}_e = \{m \in \mathcal{M}(P) : \text{loc}(m) = e\}$
    \State $\text{Sus}_{\text{final}}(e) = \max_{m \in \mathcal{M}_e} \text{Sus}_{\text{ref}}(m)$
\EndFor
\State \Return $\mathcal{R}_{\text{MBFL-FLIM}} = \text{Sort}(\mathcal{E}, \text{key}=\text{Sus}_{\text{final}}, \text{order}=\text{desc})$

\end{algorithmic}
\end{algorithm}

The algorithm begins with standard MBFL steps (Steps 1-2) for test execution and mutant generation, then introduces the FLIM Recognition (Step 3) that combines recognition prompt engineering, fine-tuning enhancement, and confidence estimation to accurately recognize interference mutants. Steps 4-5 integrate FLIM recognition results with MBFL suspiciousness calculation through the mitigation process, while Step 6 aggregates refined mutant-level suspiciousness scores to program entity level using maximum aggregation.

\section{Experimental Design}
\label{sec:experimental_design}

To evaluate the effectiveness of MBFL-FLIM, we design a systematic experimental study. This section presents our research questions, subject programs, evaluation metrics, and experimental configuration.

\subsection{Research Questions}

\textbf{RQ1: How does the fault localization effectiveness of MBFL-FLIM compare to that of baseline methods?}

RQ1 aims to evaluate the overall effectiveness of MBFL-FLIM by comparing its fault localization performance with that of two basic fault localization methods (SBFL~\cite{abreu2006evaluation} and MBFL~\cite{papadakis2015metallaxis}), three advanced dynamic feature based fault localization methods (Delta4Ms~\cite{liu2024delta4ms, du2022improving}, BLMu~\cite{wang2025systematic} and SmartFL~\cite{wu2025smartfl}) and two LLM-based FL methods (LLMAO~\cite{yang2024large} and LLMFL~\cite{wu2023large}).

\textbf{RQ2: How does MBFL-FLIM perform in multi-fault scenarios?}

RQ2 aims to investigate the robustness of MBFL-FLIM in more challenging multi-fault scenarios, where multiple faults coexist within the same program. Multi-fault localization presents additional complexity as interference effects may be amplified, and traditional methods often struggle to maintain effectiveness. This research question aims to evaluate whether MBFL-FLIM can maintain its performance when dealing with programs containing multiple faults.

\textbf{RQ3: Whether the fine-tuning and confidence estimation components can further improve the performance of FLIM Recognition in MBFL-FLIM?}

RQ3 aims to conduct an ablation study to decompose the performance gains achieved by MBFL-FLIM into its constituent components. By systematically removing the fine-tuning component and the confidence estimation component, we aim to quantify the individual contribution of each enhancement to the overall effectiveness. This analysis evaluates the four variants of MBFL-FLIM: MBFL-FLIM I (basic version with LLM-based recognition only), MBFL-FLIM II$_{F}$ (adding fine-tuning), MBFL-FLIM II$_{C}$ (adding confidence estimation), and MBFL-FLIM (incorporating both enhancements), providing insights into which components are most critical for achieving optimal performance.

\textbf{RQ4: How do different confidence estimation algorithms impact the performance of MBFL-FLIM?}

RQ4 aims to investigate the effectiveness of various confidence estimation algorithms within the confidence estimation component. We aim to compare three independent algorithms: entropy-based weighting (EBW), count-based weighting (CBW), and principal component analysis (PCA) approaches to determine which algorithm provides the most effective confidence measurement for FLIM recognition results. This analysis helps determine the optimal strategy for measuring decision reliability and generating stable confidence scores from multiple LLM inference runs.

\begin{table}[htb]
\centering
\renewcommand{\arraystretch}{1.1}
\caption{Statistics of Subject Programs} 
\label{table-data}
\begin{tabular}{lcccc}
\toprule
\multicolumn{1}{c}{\multirow{2}{*}{\textbf{Subject}}}
& \multicolumn{1}{c}{\multirow{2}{*}{\textbf{\#Versions}}}
& \multicolumn{3}{c}{\multirow{1}{*}{\textbf{Avg.}}}
\\ \cmidrule{3-5}
& & \textbf{\#TestCases} & \textbf{\#kLOC} & \textbf{\#Mutants} \\
\midrule
Chart       & 26  & 1,814 & 203 & 2,205 \\
Closure     & 133 & 7,027 & 139 & 7,927 \\
Lang        & 65  & 1,815 &  52 & 2,245 \\
Math        & 106 & 2,513 & 116 & 3,602 \\
Mockito     & 38  & 1,140 &  19 & 1,366 \\
Time        & 27  & 3,918 &  68 & 4,130 \\
\midrule
Total       & 395 & 18,227 & 597 & 21,475 \\
\bottomrule
\end{tabular}

\end{table}

\subsection{Subject Programs}

We conduct our experiments using the Defects4J benchmark~\cite{just2014defects4j}, a widely recognized and extensively utilized collection of real-world faults from major open-source Java applications. Defects4J has been specifically designed to support software engineering research and has been employed across various domains including test case generation, fault localization, automated program repair, and regression testing. To facilitate comparison with previous research~\cite{yang2024large}, we select Defects4J v1.2.0 the evaluation dataset version. Table~\ref{table-data} presents detailed statistics for each subject program, including the number of versions, test cases, functions, and mutants. The dataset collectively encompasses 395 program versions with 21,333 test cases, 26,904 functions, and 21,475 mutants, providing a diverse evaluation foundation. This substantial dataset ensures robust statistical analysis and enables thorough assessment of MBFL-FLIM's performance across different program characteristics and fault types.

\subsection{Evaluation Metrics}
\label{subsec:Evaluation Metrics}

To evaluate the effectiveness of fault localization techniques, we employ a set of widely adopted metrics that capture different aspects of localization performance. These metrics provide both absolute and relative measures of effectiveness, enabling thorough comparison across different approaches.

\textbf{Top-N}~\cite{zou2019empirical}: The Top-N metric quantifies the number of faults successfully localized within the top N ranks of the suspiciousness-ordered list. This metric reflects the practical utility of fault localization techniques, as developers typically examine only a limited number of highly-ranked program entities. Based on empirical studies indicating that developers commonly inspect only the top-ranked entities, we evaluate Top-1, Top-3, and Top-5 performance to capture both precision and practical applicability.

\textbf{MFR (Mean First Rank)}~\cite{li2017transforming}: For programs containing multiple faulty components, identifying the location of the first faulty entity is of paramount importance for debugging efficiency. MFR measures the average rank of the first faulty entity across all evaluated faults, providing insight into how quickly developers can locate at least one fault. The metric is calculated as:

\begin{equation}
    \label{MFR}
    MFR = \frac{1}{N} \sum_{i=1}^{N} \text{Rank}_{i1}
\end{equation}

where $N$ represents the total number of faults and $\text{Rank}_{i1}$ denotes the rank of the first faulty entity for the $i$-th fault.

\textbf{EXAM (Examination Score)}~\cite{pearson2017evaluating, liu2018optimal}: The EXAM metric quantifies the proportion of program entities that developers must examine before locating the faulty component. This metric provides a normalized measure of debugging effort, where lower values indicate more efficient fault localization. EXAM is defined as:

\begin{equation}
    \label{EXAM}
    EXAM = \frac{rank}{number\ of\ executable\ statements}
\end{equation}

where the numerator represents the rank of the faulty statement and the denominator is the total number of executable statements in the program.

\subsection{MBFL-FLIM Configuration}

\subsubsection{Model Selection and Architecture}

For the FLIM recognition component of MBFL-FLIM, we carefully select a diverse set of state-of-the-art LLMs representing different architectural approaches and parameter scales. Our model selection encompasses four major series:

\textbf{Qwen and Qwen-Coder Series}: We employ both general-purpose models (Qwen-7B, Qwen-14B) and code-specialized models (Qwen-Coder-7B, Qwen-Coder-14B) from Alibaba Cloud's Qwen series. The Qwen series demonstrates strong multilingual and multimodal capabilities, with the code-specialized models specifically optimized for programming-related tasks across 92+ programming languages.

\textbf{DeepSeek Series}: We utilize three models from the DeepSeek-R1 series, including DeepSeek-R1-7B and DeepSeek-R1-14B (distilled from Qwen models), and DeepSeek-R1-Llama-8B (distilled from Llama). The DeepSeek-R1 series focuses on advanced reasoning capabilities and demonstrates superior performance on complex logical reasoning benchmarks.

\textbf{Llama Series}: We include Llama-8B from Meta's open-source series, which provides strong baseline language processing capabilities and serves as a representative of the widely-adopted Llama architecture.

Table~\ref{table:llm_models} provides detailed specifications for each selected model, including parameter counts, architectural characteristics, and optimization focus areas.

\begin{table*}[!ht]
\centering
\caption{LLMs for MBFL-FLIM}
\begin{tblr}{
  colspec = {l l l c X[l]},
  row{1} = {font=\bfseries},
  vspan=even,
  width = \textwidth,
}
\toprule
Affiliation & Model Series & Model ID & Size & Model Description \\
\midrule
\SetCell[r=4]{} Qwen & \SetCell[r=2]{} Qwen & Qwen-7B & 7B & \SetCell[r=2]{} {Qwen is Alibaba Cloud's LLM series with strong multilingual and multimodal capabilities. The series spans 0.6B to 235B parameters, featuring enhanced reasoning and long-context handling in versions like Qwen2.5. This study utilizes Qwen-7B and Qwen-14B for their balanced performance.} \\ \cmidrule{3-4}
            &            & Qwen-14B & 14B & \\
\cmidrule[lr]{2-5}
\SetCell[r=2]{}       & \SetCell[r=2]{} Qwen-Coder & Qwen-Coder-7B & 7B & \SetCell[r=2]{} {Qwen-Coder is a specialized code-focused series for generation, repair, and reasoning across 92+ programming languages. With competitive performance against GPT-4o on coding benchmarks, the series ranges from 0.5B to 32B parameters. This study employs Qwen-Coder-7B and Qwen-Coder-14B models.} \\ \cmidrule{3-4}
            &            & Qwen-Coder-14B & 14B & \\
\midrule
\SetCell[r=3]{} DeepSeek & \SetCell[r=3]{} DeepSeek-R1 & DeepSeek-R1-7B & 7B & \SetCell[r=3]{} {DeepSeek-R1 is an advanced reasoning-focused model series that excels in complex logical reasoning tasks. The series demonstrates superior performance on reasoning benchmarks through innovative training methodologies. This study utilizes distilled models of DeepSeek-R1-7B (distilled from Qwen), DeepSeek-8B(distilled from Llama), and DeepSeek-R1-14B(distilled from Qwen).} \\ \cmidrule{3-4}
            &            & DeepSeek-R1-Llama-8B & 8B & \\ \cmidrule{3-4}
            &            & DeepSeek-R1-14B & 14B & \\
\midrule
Meta        & Llama      & Llama-8B & 8B & Llama is Meta's open-source LLM series that has evolved into multimodal AI frameworks with safety features and multilingual support. The series ranges from 1B to 2T parameters, with Llama 3+ versions offering enhanced tokenization and grouped query attention. This study employs Llama-8B for balanced language capabilities. \\
\bottomrule
\end{tblr}
\vspace{0.5em}

\label{table:llm_models}
\footnotesize{\textit{Note: Model IDs represent the specific LLM used in the experiments.}}
\end{table*}

\subsubsection{Tuning Experiment for LLM Selection}

To determine the optimal LLM for MBFL-FLIM, we conduct a tuning experiment across eight representative models from different architectural series and parameter scales. This experiment evaluates the performance of MBFL-FLIM with different LLMs to select the most suitable model for subsequent experiments. Due to the computational overhead of fine-tuning, we only evaluate the performance of different LLMs on MBFL-FLIM I (the basic prompt version) and assume that the optimal model identified through this experiment will remain optimal for fine-tuning and confidence estimation variants. 

\begin{table}[htb]
\centering
\renewcommand{\arraystretch}{1.1}
\caption{Performance Comparison of MBFL-FLIM I with Different LLMs}
\label{tuning-table}
\begin{tabular}{@{}lcccc@{}}
\toprule
\textbf{Model} & \textbf{Top-1} & \textbf{Top-3} & \textbf{Top-5} & \textbf{MFR} \\
\midrule
Qwen-7B & 74 & 98 & 123 & 28.45 \\
Qwen-14B & 77 & 104 & 130 & 26.73 \\
Qwen-Coder-7B & 70 & 93 & 112 & 31.28 \\
Qwen-Coder-14B & 74 & 97 & 114 & 29.84 \\
DeepSeek-R1-7B & 74 & 101 & 126 & 27.91 \\
DeepSeek-R1-14B & 81 & 115 & 142 & 24.89 \\
DeepSeek-R1-Llama-8B & 60 & 72 & 103 & 38.52 \\
Llama-8B & 63 & 73 & 99 & 36.47 \\
\bottomrule
\end{tabular}
\end{table}

Table~\ref{tuning-table} shows the performance comparison results. DeepSeek-R1-14B achieves the best performance with 81, 115, and 142 for Top-1, Top-3, and Top-5 respectively, and the lowest MFR of 24.89. The results indicate that: (1) larger parameter models generally perform better within the same architectural family; (2) DeepSeek-R1-14B outperforms other models across all metrics. Based on the experimental results, DeepSeek-R1-14B is selected as the LLM for MBFL-FLIM in subsequent experiments.

\subsubsection{Fine-tuning Configuration}

To enhance FLIM recognition accuracy, we implement domain-specific fine-tuning using a leave-one-out cross-validation approach within individual projects. This intra-project training strategy ensures that models are adapted to project-specific characteristics while maintaining generalizability across different fault types.

\textbf{Dataset Construction}: The fine-tuning dataset is constructed from the Defects4J benchmark, where mutants killed by at least one failing test but not on the faulty statements are labeled as FLIMs. Each training instance includes the original code, mutant code differences, error messages, stack traces, and mutation location information, providing comprehensive context for FLIM recognition.

\textbf{Training Strategy}: We employ LoRA (Low-Rank Adaptation) for efficient fine-tuning to reduce computational overhead while maintaining model effectiveness. The training process incorporates early stopping mechanisms to prevent overfitting and employs appropriate regularization techniques to enhance generalization across different projects and fault scenarios.

\textbf{Validation Strategy}: We implement a 90\%-10\% train-validation split within each project, using accuracy as the primary evaluation metric. The leave-one-out approach ensures that each project serves as both training and testing data, providing robust evaluation of the fine-tuning effectiveness.

\subsubsection{Confidence Estimation Algorithms}

To address the inherent instability of LLM outputs in FLIM recognition, we implement three independent confidence estimation algorithms that operate on multi-decision matrices constructed from multiple inference runs:

\textbf{Entropy-Based Weighting (EBW)}: This algorithm leverages information entropy to measure uncertainty across different inference runs. For each mutant, it computes the probability distribution over binary decisions and calculates entropy to derive confidence scores. Higher entropy indicates lower confidence, with the final confidence score computed as $\phi_{\text{EBW}}(m_i) = 1 - \frac{H_i}{\log 2}$, where $H_i$ represents the information entropy of the decision distribution.

\textbf{Count-Based Weighting (CBW)}: This method provides an intuitive confidence measure by directly using the proportion of positive FLIM recognition decisions across multiple inference runs. The confidence score is calculated as $\phi_{\text{CBW}}(m_i) = \frac{1}{K}\sum_{j=1}^{K}R_{i,j}$, where $K$ represents the number of inference runs and $R_{i,j}$ denotes the binary decision result.

\textbf{Principal Component Analysis (PCA)}: This technique applies dimensionality reduction to the multi-decision matrix while preserving the most informative variance. It computes the first principal component of the decision matrix and projects each mutant's decision vector onto this component, followed by min-max normalization to obtain confidence scores in the $[0,1]$ interval.

\subsection{Baseline Methods}
\label{subsec:Baseline Methods}

To evaluate MBFL-FLIM's effectiveness, we compare it against representative methods from three categories of fault localization approaches:

\textbf{Traditional Methods}: We include SBFL~\cite{abreu2006evaluation} and MBFL~\cite{papadakis2015metallaxis} as fundamental baselines that represent the core techniques in fault localization research.

\textbf{Dynamic feature-based Methods}: We evaluate against recent MBFL improvements including Delta4Ms~\cite{liu2024delta4ms, du2022improving}, which addresses mutant bias through signal theory; BLMu~\cite{wang2025systematic}, which assigns weights based on kill types; and SmartFL~\cite{wu2025smartfl}, which incorporates dynamic features for enhanced effectiveness.

\textbf{LLM-based Methods}: We compare with LLMAO~\cite{yang2024large}, a test-free fault localization approach leveraging LLMs for code analysis, and LLMFL~\cite{wu2023large}, an empirical study of ChatGPT's fault localization capabilities.

Detailed technical descriptions and comparative analysis of these methods are provided in Section~\ref{sec: Related Work}.

\subsection{Experimental Environment}

All experiments are conducted on a standardized computing environment featuring Ubuntu 22.04 LTS operating system and NVIDIA A100 Tensor Core GPU with 80GB memory capacity. We utilize Python 3.9 for implementation, with GZoltar~\cite{campos2012gzoltar} for spectrum collection and Major~\cite{just2014defects4j} for mutation analysis. The Ochiai~\cite{abreu2006evaluation} suspiciousness formula is employed consistently across all experiments to ensure fair comparison.

The experimental setup ensures reproducibility through fixed random seeds and consistent environmental configurations. Each experiment is repeated multiple times to account for stochastic variations in LLM outputs, with statistical aggregation applied to obtain robust performance estimates.

\section{Results Analysis}
\label{sec: Results Analysis}

\subsection{RQ1: How does the fault localization effectiveness of MBFL-FLIM compare to that of baseline methods?}

RQ1 evaluates the overall effectiveness of MBFL-FLIM by comparing its fault localization performance against seven representative baseline methods across the entire Defects4J dataset. The baseline methods encompass four main technical approaches: (1) traditional SBFL that analyzes test execution coverage patterns; (2) MBFL approaches including classical MBFL and enhanced techniques (Delta4Ms, BLMu) that leverage mutation analysis with improved suspiciousness calculation or mutant selection strategies; (3) statistical analysis-based fault localization methods (SmartFL) that employ statistical techniques for fault recognition; and (4) LLM-based fault localization methods (LLMAO, LLMFL) that utilize semantic analysis capabilities for fault recognition. This comparison employs three core evaluation metrics: Top-N (measuring the number of faults localized within the top N suspicious statements), MFR (Mean First Rank, indicating the average rank of the first fault-revealing statement), and EXAM (measuring the percentage of code examined before locating the fault).

\begin{table}[htbp]
\centering
\renewcommand{\arraystretch}{1.1}
\caption{The TOP-N and MAR Comparison between MBFL-FLIM and Baselines}
 \label{RQ1-table}
\begin{tabular}{@{}lcccccc@{}}
\toprule
\textbf{Method} & \textbf{Top-1} & \textbf{Top-3} & \textbf{Top-5} & \textbf{MFR}  \\ 
\midrule

SBFL           & 37             & 88            & 127            & 51.85                     \\
MBFL           & 35             & 80            & 126            & 54.67                      \\
Delta4Ms        & 70             & 108           & 134            & 27.01                      \\
BLMu           & 55             & 99            & 123            & 35.24                    \\
SmartFL           & 50             & 93            & 115      & 42.73                    \\
LLMAO         & 82             & 116            & 145            & -                   \\
LLMFL        & 75             & 102            & 138             &   -                 \\
MBFL-FLIM& 96            & 133            & 171            &   19.74                  \\
\bottomrule
\end{tabular}
\end{table}

Table~\ref{RQ1-table} demonstrates that MBFL-FLIM achieves substantial improvements across all evaluation metrics. Specifically, MBFL-FLIM successfully localizes 96 faults within the top-1 suspicious statement, significantly outperforming all baseline methods with an average absolute improvement of 44 faults. For the Top-3 metric, MBFL-FLIM localizes 133 faults, representing an average absolute improvement of 30 faults over all baseline methods. The Top-5 results further validate MBFL-FLIM's effectiveness, with 171 successfully localized faults and an average absolute improvement of 38 faults. Additionally, MBFL-FLIM achieves the lowest MFR of 19.74, indicating superior ranking effectiveness.

\begin{figure}[htbp]
    \centering
    \centering\includegraphics[scale = 0.34]{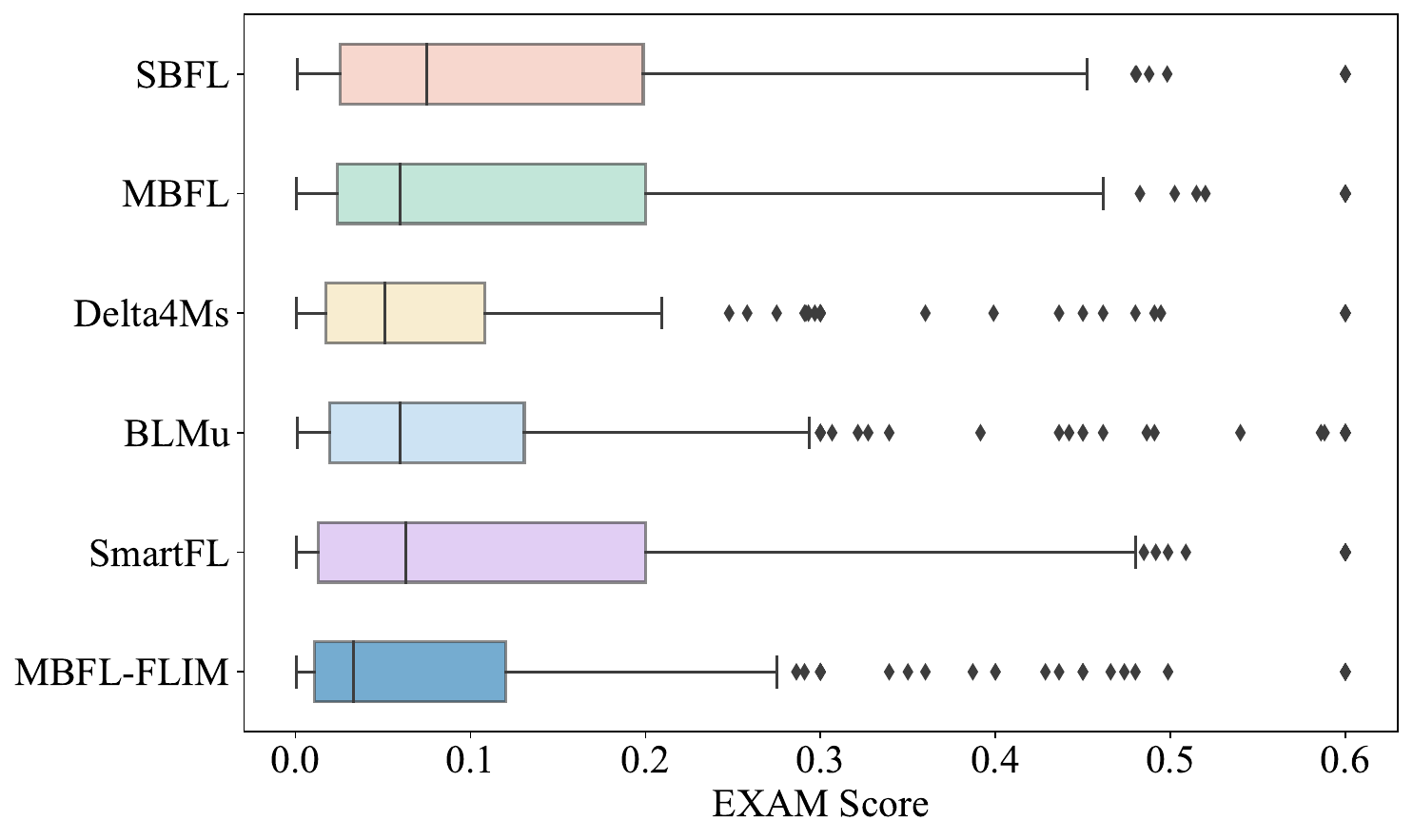}
    \caption{The EXAM Distribution Comparison between MBFL-FLIM and Baselines}
    \label{RQ1-plot}
\end{figure}

Figure~\ref{RQ1-plot} presents the EXAM score distribution, revealing that MBFL-FLIM consistently outperforms all baseline methods with a notably lower median EXAM score and an overall distribution closer to the y-axis, indicating that the method requires examining less code before locating faults. The boxplot analysis demonstrates that MBFL-FLIM requires examining substantially less code before locating faults compared to all baseline approaches, with particularly notable improvements over traditional methods (SBFL, MBFL) and competitive advantages over recent LLM-based approaches (LLMAO, LLMFL). The improved efficiency can be attributed to MBFL-FLIM's dual-component architecture: the FLIMs Recognition component effectively identifies interference mutants that mislead traditional fault localization techniques, while the Suspiciousness Refinement component appropriately scales down the suspiciousness scores of identified FLIMs, thereby enhancing the prominence of real fault-revealing statements and reducing the amount of code developers need to examine.

The experimental results demonstrate that MBFL-FLIM, as an advanced MBFL approach, achieves notable improvements over existing baseline methods in fault localization scenarios. The significant improvements over classical MBFL validate the effectiveness of addressing FLIMs through LLM-based semantic analysis. Compared to recent LLM-based fault localization methods, MBFL-FLIM's performance demonstrates the advantages of integrating dynamic execution information with semantic understanding, combining the precision of mutation-based analysis with the reasoning capabilities of LLMs, which proves particularly beneficial for large-scale and complex software projects where concrete runtime behavior analysis becomes crucial for accurate fault recognition.

\begin{center}
\begin{tcolorbox}[colback=gray!10,
                  colframe=black,
                  arc=1mm, auto outer arc,
                  boxrule=0.5pt
                 ]
\textbf{Summary for RQ1:} 

MBFL-FLIM demonstrates better fault localization performance over baseline methods across all evaluation metrics, achieving 96 Top-1 localizations (average absolute improvement of 44 faults over baseline methods), the lowest MFR of 19.74, and better EXAM performance with reduced code examination requirements, validating the effectiveness of FLIMs recognition and suspiciousness refinement in improving fault localization effectiveness.

\end{tcolorbox}
\end{center}

\subsection{RQ2: How does MBFL-FLIM perform in multi-fault scenarios compared to baseline methods?}

RQ2 investigates the robustness and effectiveness of MBFL-FLIM in multi-fault scenarios, which present significantly greater complexity due to the presence of multiple concurrent faults that can amplify interference effects and challenge traditional fault localization approaches. To conduct this evaluation, we filtered the Defects4J dataset to include only programs containing multiple faults, creating a more challenging experimental setting. The same baseline methods (SBFL, MBFL, Delta4Ms, BLMu, SmartFL, LLMAO, LLMFL) and evaluation metrics (Top-N, MFR, EXAM) are employed to ensure consistent comparison with RQ1 results.

\begin{table}[htbp]
\centering
\renewcommand{\arraystretch}{1.1}
\caption{The TOP-N and MAR Comparison between MBFL-FLIM and Baselines (RQ2)}
 \label{RQ2-table}
\begin{tabular}{@{}lcccccc@{}}
\toprule
\textbf{Method} & \textbf{Top-1} & \textbf{Top-3} & \textbf{Top-5} & \textbf{MFR}  \\ 
\midrule

SBFL           & 13             & 59            & 94            & 91.85                     \\
MBFL           & 19             & 66            & 91            & 80.29                      \\
Delta4Ms        & 51             & 82            & 99            & 47.27                      \\
BLMu           & 37             & 73            & 90            & 58.83                    \\
SmartFL           & 28             & 60            & 78      & 82.73                    \\
LLMAO         & 67             & 91            & 105            & -                   \\
LLMFL        & 60             & 88            & 103            & -                 \\
MBFL-FLIM& 75            & 100            & 114            & 28.58                  \\
\bottomrule
\end{tabular}
\end{table}

Table~\ref{RQ2-table} reveals that MBFL-FLIM maintains its superior performance in the more challenging multi-fault scenarios. For Top-1 localization, MBFL-FLIM successfully identifies 75 faults, representing an average absolute improvement of 36 faults over all baseline methods. The performance gap becomes more pronounced in Top-3 and Top-5 metrics, where MBFL-FLIM achieves 100 and 114 localizations respectively, with average absolute improvements of 26 and 20 faults. Notably, MBFL-FLIM achieves the lowest MFR of 28.58, indicating superior ranking effectiveness even in complex multi-fault environments.

\begin{figure}[htbp]
    \centering
    \centering\includegraphics[scale = 0.32]{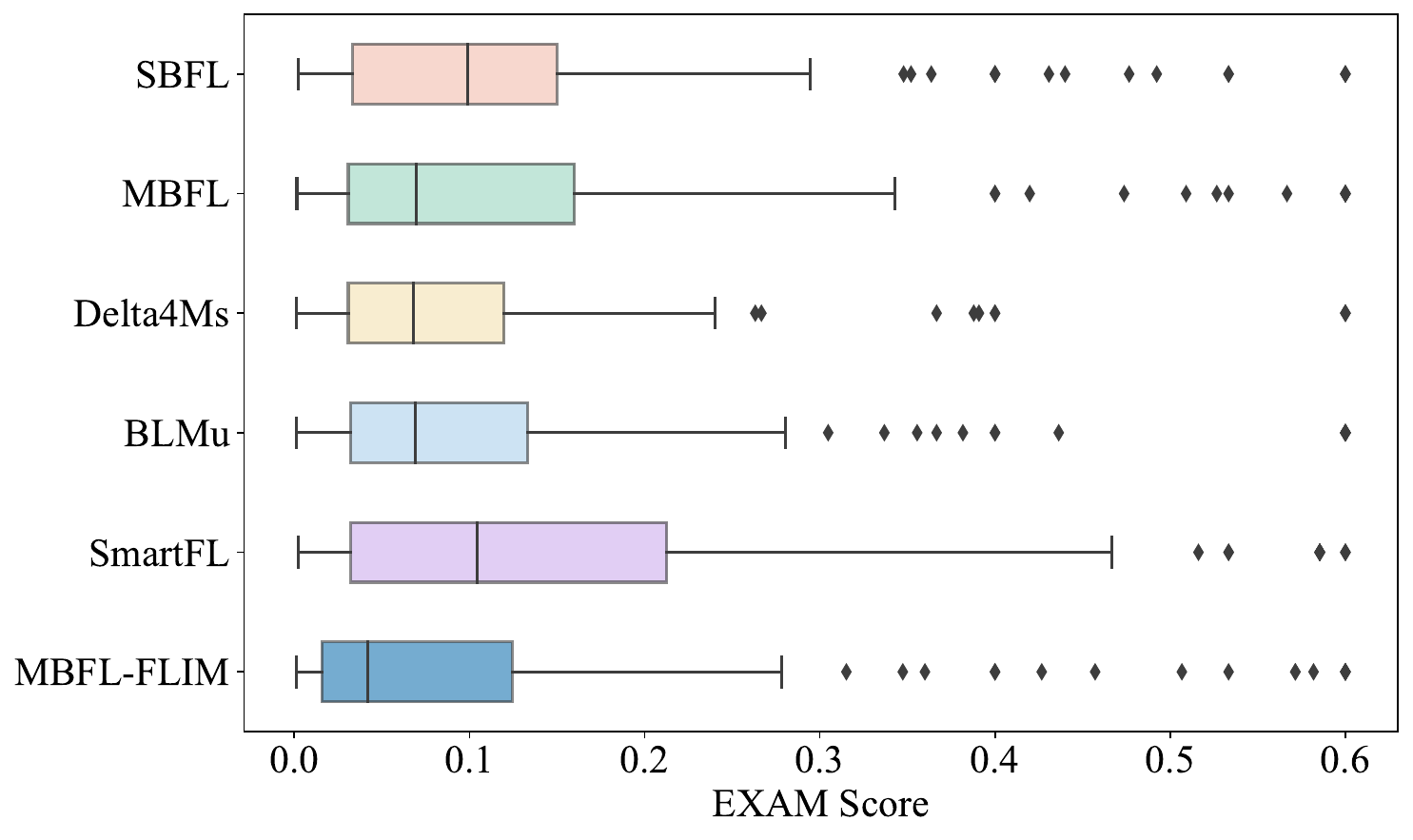}
    \caption{The EXAM Distribution Comparison between MBFL-FLIM and Baselines on Multiple Fault Scenario}
    \label{RQ2-plot}
\end{figure}

Figure~\ref{RQ2-plot} demonstrates that MBFL-FLIM's EXAM performance remains consistently superior in multi-fault scenarios, with a lower distribution compared to baseline methods. This stability is particularly significant given that multi-fault programs typically exhibit more complex fault interactions and increased interference behaviors that can mislead traditional localization techniques. The robust performance indicates that MBFL-FLIM's FLIMs recognition approach effectively handles the amplified interference effects present in multi-fault scenarios, while the confidence estimation component provides additional stability against the increased uncertainty inherent in such complex environments.

The experimental results establish that MBFL-FLIM maintains its effectiveness advantage in multi-fault scenarios, demonstrating the method's robustness. The consistent performance improvements across all metrics validate that the FLIMs-based approach successfully addresses the fundamental challenges posed by interference mutants, even under the complexity of multi-fault environments.

\begin{center}
\begin{tcolorbox}[colback=gray!10,
                  colframe=black,
                  arc=1mm, auto outer arc,
                  boxrule=0.5pt
                 ]
\textbf{Summary for RQ2:} 

MBFL-FLIM demonstrates robust performance in multi-fault scenarios, achieving 75 Top-1 localizations (average absolute improvement of 36 faults over baseline methods), the lowest MFR of 28.58 and lower EXAM distribution, validating the method's effectiveness and stability in handling complex multi-fault environments with amplified interference effects.

\end{tcolorbox}
\end{center}

\subsection{RQ3: Whether the fine-tuning and confidence estimation components can further improve the performance of FLIM Recognition in MBFL-FLIM?}

RQ3 conducts a systematic ablation study to decompose the performance gains achieved by MBFL-FLIM into the individual contributions of its constituent components. By systematically removing the fine-tuning component and the confidence estimation component, we quantify the specific enhancement provided by each component to the overall fault localization effectiveness. This analysis evaluates four variants: MBFL-FLIM I (basic LLM-based recognition only), MBFL-FLIM II$_{F}$ (adding fine-tuning), MBFL-FLIM II$_{C}$ (adding confidence estimation), and MBFL-FLIM (incorporating both enhancements), using Top-N, MFR, and EXAM metrics.

\begin{table}[htb]
\centering
\caption{Component Ablation Comparison of MBFL-FLIM on Top-N and MAR Metrics}
\label{RQ3-table}
\begin{tabular}{@{}lcccc@{}}
\toprule
\textbf{Method} & \textbf{Top-1} & \textbf{Top-3} & \textbf{Top-5} & \textbf{MFR} \\ 
\midrule
MBFL-FLIM I    & 64             & 104            & 124   & 31.32     \\
MBFL-FLIM II$_{C}$   & 81             & 115            & 142   & 24.89     \\
MBFL-FLIM II$_{F}$ & 80             & 112            & 141   & 25.14     \\
MBFL-FLIM       & 96             & 133            & 171   & 19.74     \\
\bottomrule
\end{tabular}
\end{table}

Table~\ref{RQ3-table} reveals the distinct contributions of each component to MBFL-FLIM's overall performance. The basic variant MBFL-FLIM I achieves 64, 104, and 124 for Top-1, Top-3, and Top-5 respectively, establishing the foundation performance using only basic LLM-based FLIM recognition. The addition of confidence estimation (MBFL-FLIM II$_{C}$) demonstrates substantial improvements, achieving 81 (an absolute increase of 17 faults), 115 (an absolute increase of 11 faults), and 142 (an absolute increase of 18 faults) respectively, while reducing MFR from 31.32 to 24.89. Similarly, incorporating fine-tuning (MBFL-FLIM II$_{F}$) yields significant enhancements with 80 (an absolute increase of 16 faults), 112 (an absolute increase of 8 faults), and 141 (an absolute increase of 17 faults) respectively. The MBFL-FLIM, integrating both components, achieves optimal performance with 96, 133, and 171, representing the highest absolute improvements of 32 faults for Top-1, 29 faults for Top-3, and 47 faults for Top-5 over the baseline MBFL-FLIM I approach.

\begin{figure}[htp]
    \centering
    \centering\includegraphics[scale = 0.30]{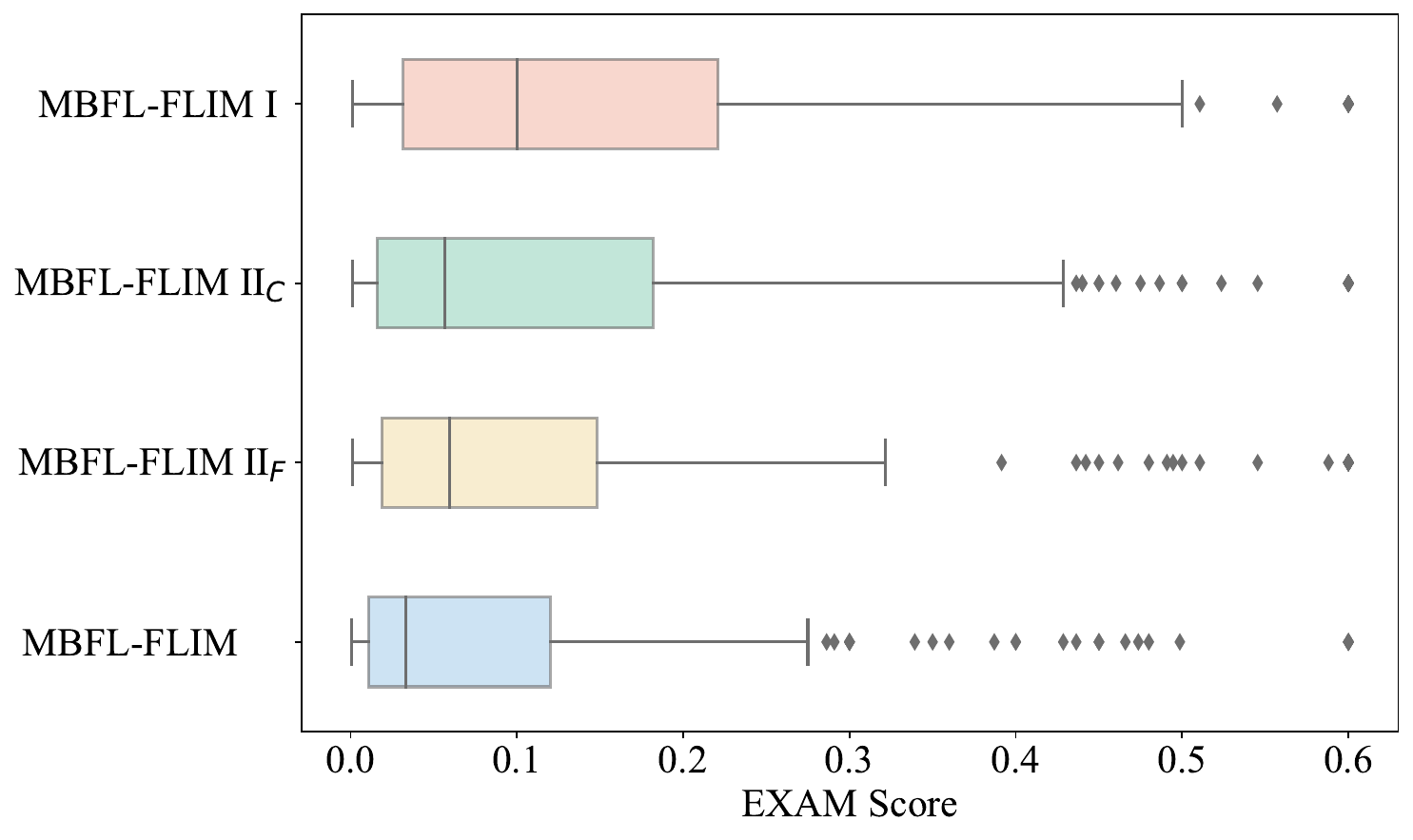}
    \caption{Component Ablation Comparison of MBFL-FLIM on EXAM Metric}
    \label{RQ3-plot}
\end{figure}

Figure~\ref{RQ3-plot} demonstrates the progressive improvement in EXAM performance as components are systematically added to the baseline approach. MBFL-FLIM I exhibits the highest EXAM distribution, indicating less efficient fault localization. The addition of either confidence estimation (MBFL-FLIM II$_{C}$) or fine-tuning (MBFL-FLIM II$_{F}$) both slightly reduces the EXAM distribution, with confidence estimation showing slightly superior performance. The MBFL-FLIM achieves the lowest median EXAM scores, demonstrating the synergistic effect of combining both components.

The ablation analysis reveals that both components contribute complementary enhancements to MBFL-FLIM's effectiveness. The fine-tuning component optimizes the LLM's parameters specifically for FLIM detection tasks, improving the model's ability to recognize subtle characteristics that distinguish interference mutants from regular mutants. The confidence estimation component addresses the inherent uncertainty in LLM predictions by aggregating multiple inference results and providing reliability scores, thereby reducing false positive recognitions and enhancing overall stability. The superior performance of MBFL-FLIM demonstrates that these components work synergistically, with fine-tuning improving the quality of individual predictions while confidence estimation ensures robust aggregation of multiple predictions.

\begin{center}
\begin{tcolorbox}[colback=gray!10,
                  colframe=black,
                  arc=1mm, auto outer arc,
                  boxrule=0.5pt
                 ]
\textbf{Summary for RQ3:} 
The ablation study demonstrates that both fine-tuning and confidence estimation components provide substantial individual contributions, with the complete MBFL-FLIM achieving optimal performance (96/133/171 for Top-1/3/5) through their synergistic combination, representing absolute improvements of 32 faults for Top-1, 29 faults for Top-3, and 47 faults for Top-5 over the baseline MBFL-FLIM I approach.

\end{tcolorbox}
\end{center}

\subsection{RQ4: How do different confidence estimation algorithms impact the performance of MBFL-FLIM?}

RQ4 investigates the effectiveness of different confidence estimation algorithms within the confidence estimation component to determine the optimal strategy. This analysis compares three distinct algorithms: Entropy-based Weighting (EBW), Count-based Weighting (CBW), and Principal Component Analysis (PCA), evaluating their respective contributions to fault localization accuracy and stability. The comparison employs Top-N, MFR, and EXAM metrics to assess how different confidence estimation algorithms influence the overall performance of MBFL-FLIM.

\begin{table}[htb]
\centering
\caption{Performance Comparison of MBFL-FLIM II$_{C}$ with Different Confidence Algorithms on Top-N and MAR Metrics}
\label{RQ4-table}
\begin{tabular}{@{}lcccc@{}}
\toprule
\textbf{Method} & \textbf{Top-1} & \textbf{Top-3} & \textbf{Top-5} & \textbf{MFR} \\ 
\midrule
MBFL-FLIM II$_{C}$ (EBW) & 77 & 114 & 132 & 25.73 \\
MBFL-FLIM II$_{C}$ (CBW) & 80 & 113 & 138 & 26.65 \\
MBFL-FLIM II$_{C}$ (PCA) & 81 & 115 & 142 & 24.89 \\
\bottomrule
\end{tabular}
\end{table}

Table~\ref{RQ4-table} demonstrates the comparative performance of different confidence estimation algorithms in aggregating multi-model predictions. The Entropy-based Weighting (EBW) approach achieves 77, 114, and 132 for Top-1, Top-3, and Top-5 respectively, with an MFR of 25.73, demonstrating solid performance by leveraging prediction uncertainty to weight model contributions. The Count-based Weighting (CBW) algorithm yields comparable results with 80, 113, and 138, achieving an MFR of 26.65, indicating that simple majority voting provides reasonable aggregation effectiveness. The Principal Component Analysis (PCA) approach demonstrates superior performance, achieving 81 (an absolute increase of 4 faults over EBW and 1 fault over CBW), 115 (an absolute increase of 1 fault over EBW and 2 faults over CBW), and 142 (an absolute increase of 10 faults over EBW and 4 faults over CBW) respectively, while attaining the lowest MFR of 24.89. Across all confidence estimation algorithms, the average performance is 79 for Top-1, 114 for Top-3, and 137 for Top-5, with PCA consistently outperforming the average by 2, 1, and 5 faults respectively.

\begin{figure}[htbp]
    \centering
    \centering\includegraphics[scale = 0.25]{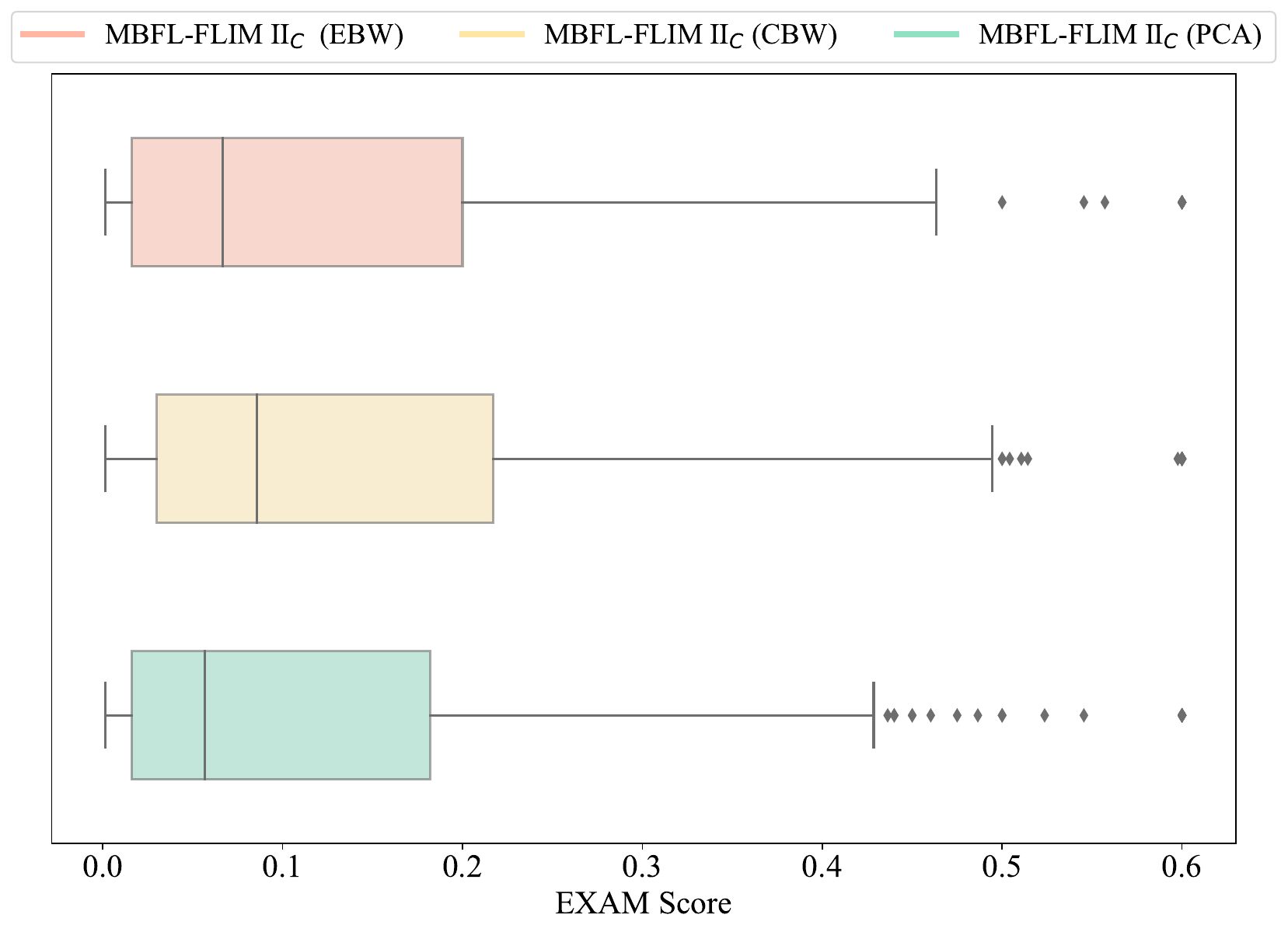}
    \caption{Performance Comparison of MBFL-FLIM II$_{C}$ with Different Confidence Algorithms on EXAM Metric}
    \label{RQ4-plot}
\end{figure}

Figure~\ref{RQ4-plot} illustrates the EXAM score distributions for different confidence estimation algorithms, revealing distinct patterns in fault localization efficiency. The EBW approach exhibits a moderate median EXAM score, demonstrating reasonable performance in aggregating model predictions based on uncertainty measures. CBW shows slightly higher median performance, indicating that simple majority voting provides effective aggregation capabilities. The PCA algorithm demonstrates superior performance with the lowest median EXAM scores, suggesting more efficient and stable fault localization across diverse program contexts.

The superior performance of PCA can be attributed to its sophisticated dimensionality reduction capabilities that capture the most informative features from multiple model predictions while filtering out noise and redundant information. Unlike EBW, which relies solely on prediction uncertainty, or CBW that treats all models equally, PCA identifies the principal components that contribute most significantly to accurate FLIM recognition. This approach effectively leverages the complementary strengths of different LLM models while mitigating their individual weaknesses, resulting in more robust and accurate fault localization performance.

\begin{center}
\begin{tcolorbox}[colback=gray!10,
                  colframe=black,
                  arc=1mm, auto outer arc,
                  boxrule=0.5pt
                 ]
\textbf{Summary for RQ4:} 
The PCA confidence estimation algorithm demonstrates superior performance across all metrics, achieving the best results with MFR of 24.89. PCA effectively captures key information from multiple FLIMs recognitions while reducing noise and enhancing aggregation stability, outperforming both EBW and CBW approaches.
\end{tcolorbox}
\end{center}

\section{Threats To Validity}
\label{sec:Threats to Validity}

\subsection{Internal Validity}

The correctness of implementation in several core components of our approach may pose internal threats to validity. Specifically, in the mutation analysis and suspiciousness calculation of the MBFL component, we rely on the official Defects4J interfaces and scripts to perform mutation analysis, ensuring standard and reproducible procedures. For suspiciousness computation, we implemented the formulas based on well-established fault localization studies and cross-validated our implementation with publicly available reference code to minimize potential errors~\cite{abreu2006evaluation,wong2013dstar}.
Additionally, the fine-tuning process of LLMs may introduce implementation-related risks. To mitigate this, we employed established fine-tuning strategies, referenced existing literature~\cite{Ding2023ParameterEfficient,10.1145/3735633}, and calibrated hyperparameters through preliminary experiments to ensure sound methodology.
Finally, the computation of confidence scores using techniques such as Principal Component Analysis (PCA)~\cite{jolliffe2016principal} and information entropy could be a source of inaccuracy if misapplied. We addressed this by leveraging established third-party libraries and validating the output, thus improving the reliability of the results.

\subsection{External Validity}

External validity concerns the extent to which our research findings can be generalized to other software systems and application scenarios. One important threat lies in the diversity of codebases, functionalities, and programming languages. Although Defects4J provides programs of varying sizes and domains, it is limited to the Java language. To enhance the generalizability of our results, we plan to extend our experiments in the future to additional datasets in C and Python, thereby evaluating the robustness of our approach across different ecosystems. 
Moreover, the generalization capability of the MBFL-FLIM method across different foundation models is another potential concern. To address this challenge, we introduced a variety of LLMs with distinct characteristics, ranging from models specialized in code processing to those with strong reasoning abilities. This ensures that our findings are not overly dependent on any specific type of model.

\subsection{Construct Validity}

Construct validity refers to how well our chosen metrics and definitions capture the intended concepts. To evaluate fault localization performance, we employed Top-N, Mean Fault Rank (MFR), and EXAM score, which are widely used in the literature and recognized for their reliability and interpretability~\cite{wong2016survey,li2017transforming}. 
Furthermore, the construct of FLIMs could be questioned in terms of how accurately it captures the concept of interference. We addressed this by grounding the definition in the RIPR model, which provides a theoretically sound framework for analyzing the impact of mutants on fault localization. This formalization also guided our prompt engineering process, making it more systematic and theoretically justified.

\section{Related Work}
\label{sec: Related Work}

\subsection{Recent Advances in MBFL Effectiveness Improvement}

Recent developments in MBFL effectiveness improvement have primarily focused on two key aspects: enhancing mutation analysis techniques and improving suspiciousness calculation methods. From the mutation analysis perspective, researchers have explored various approaches to optimize mutant generation. Neural-MBFL ~\cite{du2024neural} leverages deep learning to generate semantically meaningful mutants, achieving substantial improvements in fault localization accuracy. Neuralntegra-MBFL~\cite{liu2025integrating} combines neural and traditional mutants to enhance the diversity and quality of mutation analysis. Regarding suspiciousness calculation, several innovative methods have been developed to refine the accuracy of fault localization. PRMA~\cite{yan2023effective} combines the PageRank algorithm with mutation analysis to improve the prioritization of suspicious code entities by leveraging mutant importance within the program's structural context. Wu et al. proposed GMBFL, a graph-based technique that integrates diverse fault-related information through a unified graph representation and utilizes a gated graph attention neural network to compute suspiciousness more precisely~\cite{wu2023gmbfl}. Liu et al. presented Delta4Ms, which addresses mutant bias by modeling suspiciousness as a combination of desired information and bias information, thereby refining suspiciousness scores through statistical modeling and higher-order mutants~\cite{liu2024delta4ms, du2022improving}. Wang et al. proposed 
BLMu, which assigns different weights to mutants based on their contributions to fault localization by considering different kill types~\cite{wang2025systematic}.

While these advances have significantly improved MBFL performance, they primarily focus on either optimizing mutant generation or enhancing suspiciousness calculation algorithms. In contrast, this study introduces a fundamentally different approach by addressing the interference phenomena in mutation analysis. Rather than improving mutant generation or suspiciousness calculation directly, our work focuses on recognizing and mitigating interference mutants that mislead the fault localization process.

\subsection{LLM Applications in Fault Localization}  

LLMs have demonstrated promising capabilities in fault localization as an application domain where their semantic reasoning capabilities complement traditional spectrum-based and mutation-based approaches. 
Wu et al. conducted LLMFL~\cite{wu2023large}, an empirical study evaluating ChatGPT's fault localization capabilities, finding that ChatGPT-4 outperforms existing methods with function-level code context but shows performance degradation when using class-level code context, highlighting context-sensitivity limitations in LLM-based approaches. Yang et al. introduced LLMAO~\cite{yang2024large}, a test-free fault localization approach that leverages LLMs to analyze code without requiring test execution, demonstrating improvements over traditional spectrum-based fault localization methods. Kang et al. developed AutoFL~\cite{kang2024quantitative}, a framework that employs a two-stage approach where LLMs generate natural language explanations for potential root causes and then aggregate these explanations to predict faulty methods, achieving improvements over spectrum-based baselines while providing interpretable explanations. Ji et al. explored the application of sequence generation fine-tuning techniques in code analysis tasks~\cite{ji2024impact}, contributing to the understanding of how LLMs can be adapted for software engineering applications and demonstrating measurable improvements in fault localization performance metrics. 
However, these LLM-based fault localization approaches face challenges that hard to conduct project-level fault localization and often requiring integration with dynamic execution information to achieve optimal effectiveness~\cite{kang2024quantitative}.

MBFL-FLIM differs from existing LLM-based fault localization approaches in both scope and methodology. While fault localization methods like LLMAO~\cite{yang2024large}, AutoFL~\cite{kang2024quantitative}, and Ji et al.~\cite{ji2024impact} focus on direct fault localization through broad code analysis, MBFL-FLIM addresses a specific problem within MBFL by targeting interference mutants that can mislead existing MBFL techniques. 
MBFL-FLIM enables precise LLM integration that addresses the FLIMs recognition challenge through three components: semantic behavioral analysis for recognizing FLIMs, fine-tuning enhancement for domain-specific FLIMs recognition, and confidence estimation for mitigating LLM output instability.

\section{Conclusion}
\label{sec: Conclusion}

In this paper, we introduce the concept of FLIMs defining the mutants that interfere with the MBFL and propose a novel LLM-based FLIM recognition and suspiciousness refinement-based FLIM mitigation to alleviate their negative impact in fault localization. 
Our theoretical analysis establishes a detailed understanding of FLIMs behavior through the RIPR model, revealing four distinct interference root causes: reach-based interference (control flow alteration), infection-based interference (error state masking), propagation-based interference (error transmission rerouting), and revealability-based interference (test oracle exploitation). These root causes systematically exploit program execution contexts to create misleading diagnostic information that undermines the effectiveness of MBFL. Building upon this, we propose to utilize the LLM-based semantic analysis (incorporating basic recognition prompt, FLIM-specific fine-tuning, and confidence estimation) to recognize the FLIMs and mitigate the effect of FLIMs by suspiciousness refinement. We further develop MBFL-FLIM, a fault localization framework that integrates the FLIM recognition and mitigation into the MBFL workflow. 
Empirical experimental evaluations on the Defects4J benchmark demonstrate substantial improvements over existing approaches, with MBFL-FLIM achieving 96 Top-1 localizations (an average absolute improvement of 44 faults over baseline methods, from an average baseline of 52 to 96) and maintaining robust performance in multi-fault scenarios with 75 Top-1 localizations. The ablation studies confirm that both fine-tuning and confidence estimation contribute synergistically, with Principal Component Analysis-based confidence estimation and DeepSeek-R1-14B achieving optimal performance for FLIM recognition.

Future research directions include investigating advanced confidence estimation algorithms that can better capture nuanced relationships between different LLM predictions and exploring ensemble methods that integrate multiple confidence estimation strategies. Additionally, we plan to extend our approach to other programming languages beyond Java and investigate the integration of more advanced LLM capabilities, such as function calling and task planning, which could enable fine-grained analysis of different FLIM interference mechanisms and support more sophisticated reasoning about code semantics and mutation behaviors. These enhancements could potentially lead to even more accurate fault localization performance and broader applicability across diverse software engineering contexts.

\section*{DATA AVAILABILITY STATEMENT}
To facilitate reproducibility and future research, we make the implementation of our FLIM identification and processing components publicly available at \url{https://github.com/HengyuanLiu/FLIMs}. The experimental datasets and detailed experimental results are available upon reasonable request.

\bibliographystyle{IEEEtran}
\bibliography{IEEEabrv, references}

\begin{IEEEbiography}[{\includegraphics[width=1in,height=1.25in,clip,keepaspectratio]{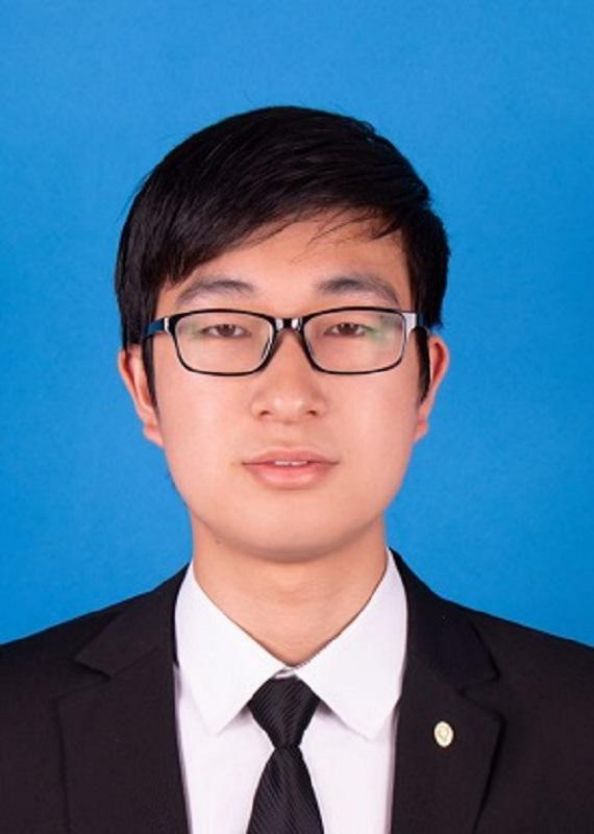}}]{Hengyuan Liu}
is currently working toward the Ph.D. degree with the College of Information Science and Technology, Beijing University of Chemical Technology (BUCT). His research interests include software testing, fault localization, and intelligent software engineering. His current research focuses on mutation-based fault localization and large language models (LLMs) applications in software engineering.
\end{IEEEbiography}

\begin{IEEEbiography}[{\includegraphics[width=1in,height=1.25in,clip,keepaspectratio]{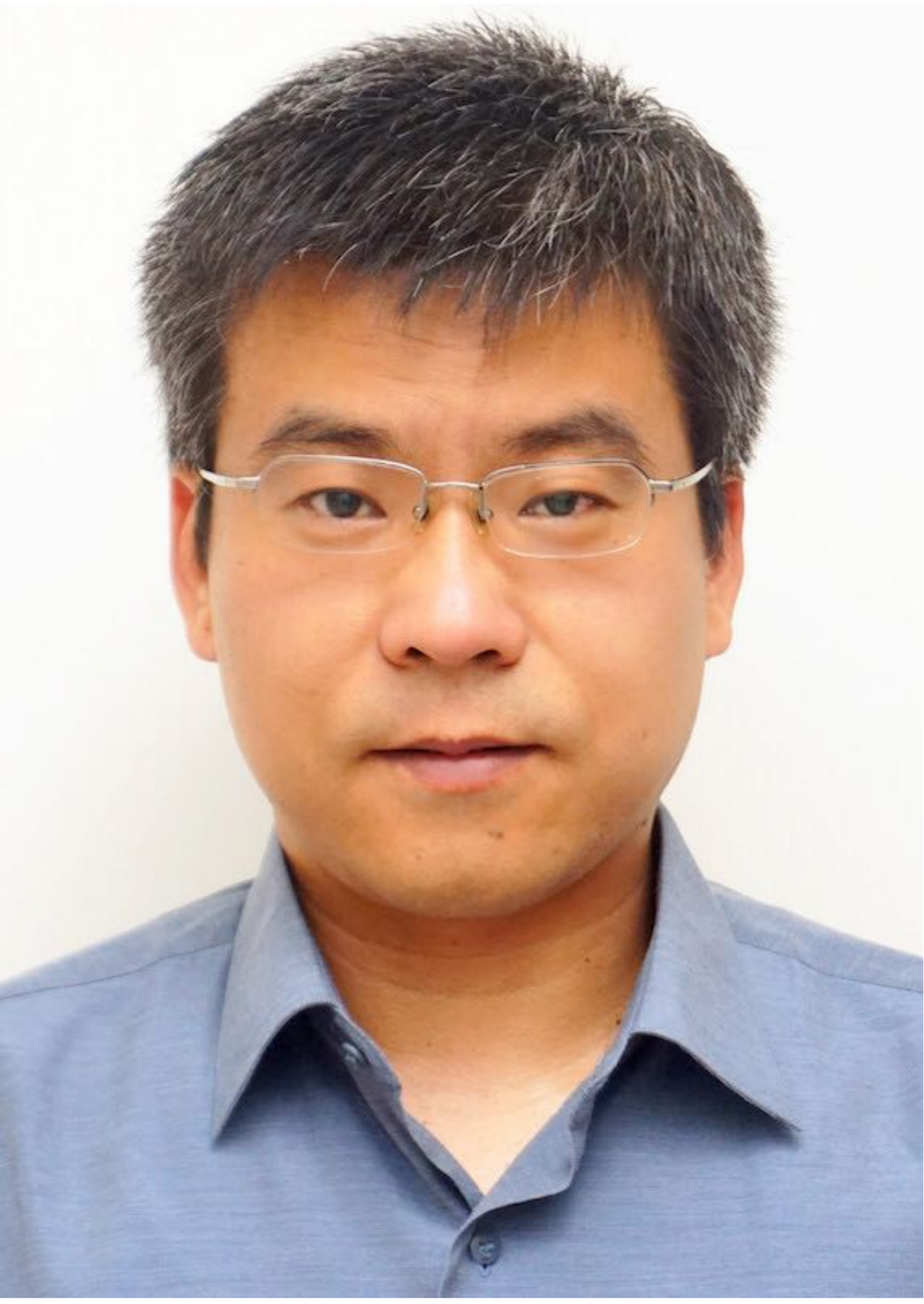}}]{Zheng Li}
received the Ph.D. degree from King's College London, CREST Centre, in 2009, under the supervision of Mark Harman. He is a Full Professor of computer science with the Department of Computer Science, Beijing University of Chemical Technology. He is the author of over 100 publications, a Co-Guest Editor for three journal special issues and has served on 20 programme committees. He has worked on program testing, source code analysis and manipulation, and intelligent software engineering.
\end{IEEEbiography}

\begin{IEEEbiography}[{\includegraphics[width=1in,height=1.25in,clip,keepaspectratio]{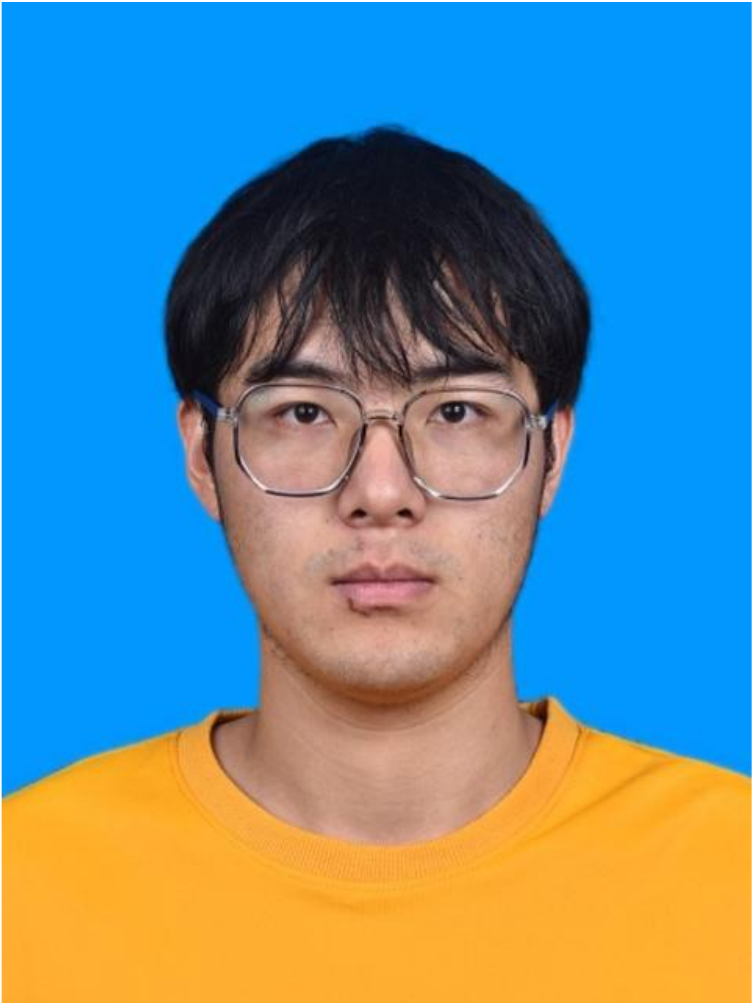}}]{Donghua Wang}
is currently working toward a master's degree at the College of Information Science and Technology, Beijing University of Chemical Technology (BUCT). His research interests lie in the fields of software engineering and software testing, with a particular focus on mutation-based fault localization and the application of large language models (LLMs) in code analysis.
\end{IEEEbiography}

\begin{IEEEbiography}[{\includegraphics[width=1in,height=1.25in,clip,keepaspectratio]{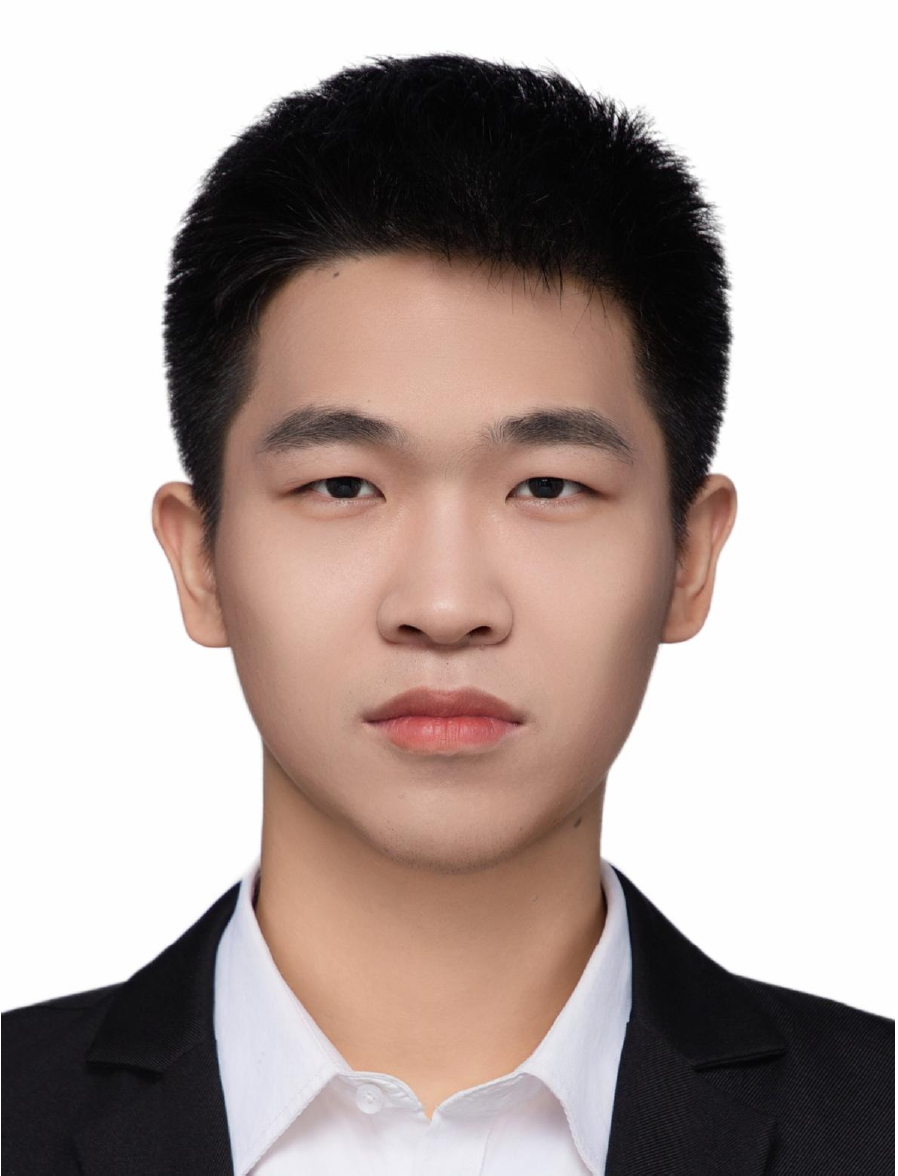}}]{Yankai Wu}
is currently working toward a master's degree at the College of Information Science and Technology, Beijing University of Chemical Technology (BUCT). His primary research focuses on competition-level program code generation methods based on large language models.
\end{IEEEbiography}

\begin{IEEEbiography}[{\includegraphics[width=1in,height=1.25in,clip,keepaspectratio]{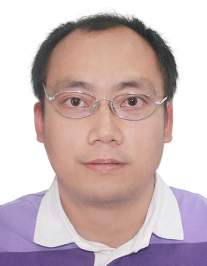}}]{Xiang Chen}
received the B.Sc. degree in the school of management from Xi'an Jiaotong University, China in 2002. Then he received his M.Sc., and Ph.D. degrees in computer software and theory from Nanjing University, China in 2008 and 2011 respectively. He is currently an Associate Professor at the School of Artificial Intelligence and Computer Science, Nantong University, Nantong, China. He has authored or co-authored more than 120 papers in refereed journals or conferences (such as TSE, TOSEM, EMSE, JSS, IST, ICSE, ESEC/FSE, ASE). His research interests include software testing and maintenance, software repository mining, and empirical software engineering. He received two ACM SIGSOFT distinguished paper awards in ICSE 2021 and ICPC 2023. He is the editorial board member of Information and Software Technology.
\end{IEEEbiography}

\begin{IEEEbiography}[{\includegraphics[width=1in,height=1.25in,clip,keepaspectratio]{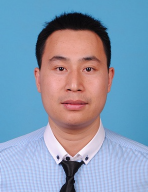}}]{Yong Liu}
(Member, IEEE) received the B.Sc. and M.Sc. degrees in computer science and technology, and the Ph.D. degree in control science and engineering from Beijing University of Chemical Technology, China, in 2008, 2011, and 2018, respectively. He is a Professor with the College of Information Science and Technology, Beijing University of Chemical Technology. His research interests are mainly in intelligent software engineering and AI applications. Particularly, he is interested in software debugging and software testing, such as source code analysis, mutation testing, fault localization, LLM4SE, and AI4SE. In these areas, he has published more than 50 papers in referred journals or conferences, such as IEEE TRANSACTIONS ON SOFTWARE ENGINEERING, IEEE JOURNAL OF SOLID-STATE CIRCUITS, IEEE TRANSACTIONS ON RELIABILITY, STVR, IST, ICSME, FSE, ASE, ISSRE, QRS, SATE, and COMPSAC. He is a member of CCF in China and ACM.
\end{IEEEbiography}

\end{document}